\documentclass[acmsmall,screen,nonacm]{acmart}
\usepackage{algorithm}
\usepackage{algpseudocode}
\usepackage{amsmath}

\title{Scalable Floating-Point Satisfiability via Staged Optimization}
\author{Yuanzhuo Zhang}
\affiliation{%
  \institution{Virginia Tech}
  \country{USA}
}

\author{Zhoulai Fu}
\affiliation{%
  \institution{State University of New York}
  \country{ Korea}
}
\affiliation{%
  \institution{Virginia Tech}
  \country{USA}
}
\affiliation{%
  \institution{Stony Brook University}
  \country{USA}
}

\author{Binoy Ravindran}
\affiliation{%
  \institution{Virginia Tech}
  \country{USA}
}

\usepackage[utf8]{inputenc}
\usepackage{upquote}
\usepackage{multirow}
\usepackage{float}
\usepackage{newunicodechar}
\newunicodechar{φ}{\ensuremath{\varphi}}
\newunicodechar{ }{~}

\usepackage[font=small,skip=4pt]{caption}
\usepackage{xspace}

\usepackage{enumitem}

\setlist[itemize]{leftmargin=*, label=--}

\setlist[enumerate]{label=\arabic*., leftmargin=*}

\usepackage{enumitem}

\newcommand{\myheading}[1] {\par\vspace{0.1em}\noindent\textbf{\textit{#1}}~}

\begin{document}

\begin{abstract}
   
This work introduces \emph{StageSAT}, a new approach to solving floating-point satisfiability that bridges SMT solving with numerical optimization. StageSAT reframes a floating-point formula as a series of optimization problems in three stages, each with increasing precision. It begins with a fast, projection-aided descent objective to efficiently guide the search toward a feasible region, then proceeds to bit-level accuracy with ULP$^2$ optimization and a final $n$-ULP lattice refinement to ensure correctness.
 By construction, the final stage uses a representing function that evaluates to zero if and only if a candidate satisfies all constraints. Thus, whenever optimization drives the final-stage objective to zero, the resulting assignment is a valid solution, providing a built-in guarantee of soundness (no spurious SAT results).
 To further improve the search, StageSAT introduces a partial monotone descent property on linear constraints via an orthogonal projection technique, which prevents the optimizer from stalling on flat or misleading objective landscapes. Critically, this solver requires no heavy bit-level reasoning or specialized abstractions of floating-point arithmetic; it treats complex arithmetic as a black-box, using runtime evaluations to navigate the input space.

We implement StageSAT and evaluate it on extensive benchmarks, including the SMT-COMP'25 floating-point suites and difficult cases from prior work. In our experiments, StageSAT proved both more scalable and more accurate than state-of-the-art optimization-based alternatives. It solved strictly more formulas than any competing solver under the same time budget -- in fact, StageSAT found most of the satisfiable instance in our benchmarks and never produced a spurious model for an unsatisfiable formula. This amounts to 99.4\% recall on satisfiable cases with 0\% false SAT in our benchmarks, exceeding the reliability of prior optimization-based solvers we tested. StageSAT also delivered significant speedups (often 5--10$\times$ faster) over traditional bit-precise SMT solvers and earlier numeric solvers. These results demonstrate that our staged optimization strategy can significantly improve both the performance and correctness of floating-point satisfiability solving. To facilitate reproduction, we provide an anonymized artifact (implementation, benchmarks, and evaluation results) as supplementary material in the submission system.

\end{abstract}

\maketitle

\section{Introduction}

\textbf{Optimization-Based Floating-Point Solving:} Traditional SMT solvers for floating-point (FP) arithmetic often struggle with scalability, especially on formulas with non-linear or transcendental operations. An alternative approach reduces a floating-point satisfiability problem to a numerical optimization problem. In this paradigm, a floating-point formula (a set of constraints) is transformed into a numeric objective function R(x) that acts as a \emph{distance to satisfaction}. This objective is non-negative and evaluates to zero if and only if the candidate assignment x is a true solution (i.e., satisfies all constraints). Deciding satisfiability then becomes an optimization task: minimize R(x) and check whether the minimum value reaches zero. Pioneering solvers like {XSat\cite{DBLP:conf/cav/FuS16}}, {goSAT\cite{DBLP:conf/fmcad/KhadraSK17}}, and {Grater\cite{grater}} demonstrated the promise of this strategy by encoding FP constraints into such objective functions and then applying global optimization techniques to find a zero. The key appeal of these \emph{optimization-based} solvers is that they can treat complex arithmetic as black-box functions—they do not need to bit-blast or explicitly enumerate the intricacies of IEEE-754 semantics. Instead, the solver is guided by numeric feedback from R(x): as the assignment x moves closer to satisfying the formula, R(x) decreases, and a solution is found when R(x) hits 0. This enables reasoning about arbitrary mathematical functions (e.g., transcendentals like sin or cos) as long as those functions can be evaluated, a notable advantage over conventional SMT methods.

\noindent\textbf{Challenges:} Early optimization-based FP solvers revealed two key challenges. First, the objective functions resulting from this reduction can be irregular and non-smooth, which hurts runtime performance and can mislead numeric search methods (gradient-based or stochastic optimizers) into local minima or flat plateaus. In practice, the solver might fail to find a solution even when one exists, simply because the search gets “stuck” in a region where the objective doesn’t meaningfully decrease. For example, {XSat} has been observed to mislabel satisfiable instances as UNSAT due to inadequate minimization (getting trapped away from the global minimum); {Grater} avoids false-UNSAT errors but often returns many \emph{timeout} results, some of which are later confirmed satisfiable.

% updated paragraph, thamks to Yuanzhuo's comments. 
Second, 
 finite precision of floating-point arithmetic and coarse step adjustments can undermine solver correctness and force a trade-off with solver efficiency. Because computing R(x) uses finite-precision arithmetic and introduces rounding errors, optimization-based  solvers like Grater employ small tolerance thresholds to decide when a formula is satisfied\cite{grater-experiment}, declaring success if R(x) < tol, even though mathematically only R(x) = 0 corresponds to a true model. However, an overly loose tolerance can yield spurious solutions, e.g. erroneously accepting x $\approx$ 0 as satisfying a constraint like x > 0 $\land$ x $\leq$ 1e-160. Conversely, XSat eschews tolerances entirely and only accepts a candidate when R(x) = 0 exactly, but it can still miss a valid solution if its step-size policy cannot nudge a candidate from x = 0 to a tiny non-zero value, e.g., x = 1e-200.

These challenges mean that, in practice, earlier optimization-based solutions often proved \emph{unscalable} and \emph{inaccurate}. They might label SAT formulas as UNSAT (or unknown) because the global minimum of a jagged objective couldn’t be found, and conversely label UNSAT formulas as SAT due to floating-point error and insufficient bit precision. Notably, XSat partially addressed the precision issue by measuring constraint violations in terms of \emph{units in the last place (ULP\cite{ulp1991})}—the smallest difference between two representable FP numbers—thus bringing bit-level rigor to its objective. However, a purely ULP-based objective landscape is highly discontinuous and difficult to search, contributing to the scalability problems.

\noindent\textbf{Our Approach: StageSAT.} We present {StageSAT}, a new floating-point SMT solver that addresses these accuracy and scalability issues via \emph{staged optimization}. The central idea is to deliberately engineer the objective function and search process in multiple phases, combining the smoothness needed for effective global search with the precision needed for correctness. StageSAT aims to be \emph{both} robust on large benchmarks \emph{and} precise in its results—reporting SAT only when a correct model is found and reporting UNSAT-GUESS (i.e., a heuristic UNSAT result without proof) only after exhaustive search attempts in the final stage suggest no model is likely to exist.

StageSAT’s design integrates several key techniques to achieve this balance:
\begin{itemize}
  \item \textbf{Multi-staged objective design:} Instead of a single monolithic objective, StageSAT incrementally refines its objective across three stages. It begins with smooth squared-error terms to quickly guide the search toward a promising region. It then transitions to a more precise 1-ULP difference measure (where a ULP, or \emph{unit in last place}, represents the smallest step in FP value) to hone in on bit-level correctness. Finally, it applies an \emph{n-ULP margin refinement\cite{weak-distance}} in the last stage, progressively tightening the allowed error margin down to zero. This staged approach preserves searchability in early phases while ultimately achieving bit-precise satisfaction—StageSAT only reports “SAT” when the final objective truly reaches zero under IEEE-754 semantics.
  \item \textbf{Orthogonal projection:} To mitigate misleading gradients in regions with linear equalities or flat surfaces, StageSAT employs an orthogonal projection technique\cite{DBLP:books/daglib/0086372}. Essentially, when optimizing constraints like \texttt{a·x + b = c} that form flat “valleys” in the objective landscape, StageSAT projects the search step orthogonally onto the constraint surface to better estimate the true distance to satisfaction. This improves the alignment between movements in the input space and decreases in the objective, acting as a practical safeguard against stalls and false convergence on plateaus.
  \item \textbf{Objective smoothing:} StageSAT applies lightweight smoothing to the objective in each stage (for example, using squared forms of error terms or tiny perturbations) to reduce erratic behavior. By “rounding off” sharp corners and eliminating zero-gradient traps early on, this smoothing makes the objective landscape easier to navigate, leading the numerical optimizer more reliably toward a global minimum.
\end{itemize}

In theory, StageSAT is \emph{sound} for satisfiability: its final objective stage is bit-precise, so it will report SAT only when a candidate solution passes a rigorous, bit-level check (R(x) = 0 exactly, meaning the assignment is a valid model). For cases where the optimizer converges to a non-zero minimum, StageSAT concludes the formula is unsatisfiable (albeit without a formal proof). While this UNSAT outcome is a heuristic “confident guess” (since StageSAT, like other numeric solvers, is incomplete), it was empirically accurate in all our evaluations. If StageSAT fails to find a solution within a given time limit, it returns \emph{unknown/timeout}.

\noindent\textbf{Results:} In practice, {StageSAT} proves both accurate and scalable. We evaluated StageSAT on the SMT-COMP’25\cite{SMTCOMP2025} Quantifier-Free Floating-Point {QF\_FP} benchmarks\cite{SMTLIB2025}, including both the “middle” and “large” benchmark suites, as well as the sets of benchmarks used by XSat, Grater and JFS\cite{jfs}. StageSAT was compared against four state-of-the-art bit-precise SMT solvers ({Z3\cite{DBLP:conf/tacas/MouraB08}}, {MathSAT\cite{mathsat5}}, {Bitwuzla\cite{DBLP:conf/cav/NiemetzP23}}, {CVC5\cite{DBLP:conf/tacas/BarbosaBBKLMMMN22}}) and four optimization-based solvers ({XSat}, {goSAT}, {Grater}, {JFS})\cite{DBLP:conf/cav/FuS16, DBLP:conf/fmcad/KhadraSK17, grater, jfs}. The results demonstrate significant advantages for StageSAT:
\begin{itemize}
  \item On the {"middle"} benchmarks (35 formulas), StageSAT solved \emph{all} 35 instances (SAT and UNSAT) with a median runtime of just 0.18 seconds.
  \item On the more challenging {"large"} benchmarks, StageSAT solved 24 of 26 SAT instances and returned unsat-guess on 21 UNSAT instances (matching known ground truth results). In total, StageSAT achieved the highest coverage, solving strictly more formulas than any competing solver due to its low timeout rate.
  \item Compared to the best SMT solvers (e.g., Bitwuzla and CVC5), StageSAT achieved up to \emph{10× faster} median runtime on satisfiable cases and avoided many timeouts, solving several formulas that those solvers could not.
  \item Compared to prior optimization-based solvers, StageSAT showed significantly improved reliability and performance. It attained \emph{99.4\% SAT recall} (missing only two satisfiable instances out of 347) and \emph{0\% false SAT} on unsatisfiable cases (never reporting a spurious model), whereas tools like goSAT and Grater often reported \emph{unknown} or suffered misclassifications. StageSAT was also about an order of magnitude faster on average than {XSat}, {goSAT}, and {Grater} in our experiments.
\end{itemize}

To understand the impact of each design decision, we performed an ablation study on StageSAT’s components (objective staging, orthogonal projection, and smoothing). Removing any one component significantly degraded performance or accuracy, confirming that each plays a critical role in the solver’s effectiveness. Although StageSAT remains incomplete (like its predecessors, it cannot produce formal UNSAT proofs), our results illustrate that numerical optimization can be a powerful complement to traditional SMT solving—especially for difficult floating-point formulas where scale and runtime are paramount.

\noindent\textbf{Contributions.} In summary, this paper makes the following contributions:
\begin{itemize}
  \item \textbf{ULP-staged objective design:} A multi-stage objective strategy that moves from a smooth squared-distance metric to a 1-ULP-accurate metric and finally an n-ULP refinement. This design achieves robust global search in early phases \emph{and} bit-precise SAT detection in the final phase, addressing both scalability and floating-point precision issues in SMT solving.
  \item \textbf{Orthogonal projection for improved descent:} A technique that applies orthogonal projection on linear constraints during optimization. This improves monotonic descent towards satisfying assignments and helps the solver avoid stagnation on misleading objective landscapes.
  \item \textbf{Smoothing for stability:} The use of lightweight objective smoothing (e.g., squaring error terms) to improve convergence in regions with discontinuities or zero-gradient traps. Smoothing yields a more stable optimization process, increasing reliability and speed.
  \item \textbf{Extensive evaluation:} A thorough experimental evaluation on SMT-COMP’25 QF\_FP benchmarks demonstrating that StageSAT outperforms prior incomplete solvers in both the number of benchmarks solved and overall runtime. StageSAT also complements the complete solvers well by producing consistent results in much less time. We also present an ablation study confirming the necessity of each proposed component.
\end{itemize}

% The rest of the paper is organized as follows: {Section 2} details the staged objective design and overall optimization flow of StageSAT. {Section 3} presents the experimental evaluation and results. {Section 4} discusses related work, and {Section 5} concludes the paper.

\section{Overview and An Illustrative Example}
\label{sec:overview}

\paragraph{Scope and goal.}
We target the \textsc{QF\_FP} (Quantifier‑Free Floating‑Point) theory. Given a CNF formula \(\mathcal{C}\) over IEEE‑754 atoms, our aim is a practically accurate satisfiability outcome:
(1) if \(\mathcal{C}\) is satisfiable, we return \textbf{sat} with a witnessing assignment \(\vec{x}^{\star}\) that validates under IEEE‑754;
(2) if optimization terminates with a strictly positive minimum, we return \textbf{unsat‑guess}—not a proof, but a confident verdict supported by our evaluation; and
(3) if the time budget expires without a decision, we return \textbf{timeout}.

\paragraph{Three‑stage overview.}
Figure~\ref{fig:overview} sketches \textsc{StageSAT}—a staged solver that balances searchability and bit‑precision, moving from smooth guidance to bit‑exactness and discrete snapping. We keep mathematics light here and refer to \S\ref{sec:theory} for precise objective definitions, weak‑distance contracts, S1’s local monotone‑descent guarantee, and S3’s discrete refinement strategy.
\begin{itemize}
  \item \textbf{S1 — projection‑aided squared objective (fast descent).}
  We start with a fast, albeit coarse optimization toward feasibility.
  We separate constraints into \(\mathcal{L}\) (linear equalities) and \(\mathcal{N}\) (nonlinear/inequalities), and build an objective of the form \(R_{\mathcal{L}} + R_{\mathcal{N}}\).
  Here \(R_{\mathcal{N}}\) is a sum of \emph{squared} residuals/violations (as in traditional optimization‑based encodings), while \(R_{\mathcal{L}}\) uses an \emph{orthogonal projection} onto the linear manifold to compute distance exactly.
  The projection fixes the non‑monotone behavior illustrated in Example~\ref{ex:mdp}, yielding a \emph{partial} monotone‑descent effect and a robust early trajectory.
  \item \textbf{S2 — ULP\(^2\) objective (bit‑level alignment).}
  We then switch to a \emph{squared ULP distance} per constraint to align the objective with IEEE‑754 semantics.
  Squaring shapes the penalty (amplifies large gaps, breaks ties), providing a stable coarse‑to‑fine descent and exposing bit‑level issues (e.g., subnormal underflow) that magnitude‑based S1 can miss.
  \item \textbf{S3 — \(n\)‑ULP refinement over the floating‑point lattice.}
  Finally, we perform a bounded, \emph{discrete} search around the S2 incumbent on the FP lattice to \emph{snap} near‑misses to exact zeros that continuous search cannot cross.
\end{itemize}

\noindent
\textit{Outcomes and guarantees.} 

StageSAT yields one of \(\{\textbf{sat},\,\textbf{unsat‑guess},\,\textbf{timeout}\}\).
A \textbf{sat} result carries a \emph{validated} model (bit‑exact under IEEE‑754).
An \textbf{unsat‑guess} is a \emph{non‑proof} verdict (positive minimum) that our experiments show to be highly consistent with complete solvers while scaling to larger instances (see \S\ref{sec:Experiments}).
A \textbf{timeout} indicates that within the allotted budget neither a model nor a stable positive minimum was reached; in practice, under the budgets in \S\ref{sec:Experiments}, this is uncommon.

\begin{figure}[htbp]
  \centering
  \fbox{\parbox{0.92\linewidth}{
  \textbf{S1 (projection‑aided squared)} \(\rightarrow\)
  \textbf{S2 (ULP\(^2\))} \(\rightarrow\)
  \textbf{S3 (\(n\)‑ULP refinement over FP lattice)}. \\
  Outcomes: \textbf{sat} (validated model) \(/\) \textbf{unsat‑guess} (positive minimum; non‑proof) \(/\) \textbf{timeout}.
  }}
  \caption{Three‑stage solving flow in \textsc{StageSAT}. S1 provides fast descent with projection on linear equalities; S2 enforces bit‑level correctness via squared ULP distance; S3 discretely refines on the floating‑point lattice. Formal objectives and guarantees appear in \S\ref{sec:theory}.}
  \label{fig:overview}
\end{figure}

\subsection{Example}
\label{ex:mdp}

\paragraph{Intuition.}
We view any constraint as a conjunction \( \mathcal{C}\equiv \mathcal{L}\wedge\mathcal{N} \), where \(\mathcal{L}\) collects the linear equalities over FP variables (considered over \(\mathbb{R}\) for geometry) and \(\mathcal{N}\) contains the remaining non‑linear and/or inequality atoms. Linear equalities are common in control and signal‑processing models.

Our design principle is intuitive: as the objective decreases, the assignment should move \emph{closer} to the model set \( \mathcal{M}=\{\vec{x}\mid \vec{x}\models\mathcal{C}\} \).  
If we could compute the Euclidean distance \(d(\vec{x},\mathcal{M})\) to the model set,
then minimizing this function to zero would solve the satisfiability problem directly and would automatically satisfy our design principle. In practice, existing optimization‑based solvers replace this ideal yet intractable distance with a proxy \(R(\vec{x})\) that satisfies two basic properties:
(i) \(R(\vec{x})\ge 0\) and (ii) \(R(\vec{x})=0 \Rightarrow \vec{x}\models \mathcal{C}\).
Such proxies enable search but need not align with geometric distance; as our toy example shows, they can violate the intended "closer \(\implies\) smaller" behavior and mislead optimization.

Our key observation is that for the linear‑equality portion of a constraint, \(\mathcal{L}={\vec{x}\mid A\vec{x}==\vec{b}}\), we can compute the exact squared distance to the solution manifold without solving the equations: orthogonally project \(\vec{x}\) to \(\mathcal{L}\) (via the Moore–Penrose pseudoinverse \cite{Penrose_1955}) and measure the squared gap to that projection. This gives a closed‑form “distance‑to‑\(\mathcal{L}\)” term that directly enforces the Monotone Descent principle on the linear part.

Our Stage 1 combines this geometric term with a simple representing function for the remaining atoms \(\mathcal{N}\) (non‑linear and/or inequalities):
$$
S_1(\vec{x}) = d(\vec{x},\mathcal{L})^2 + R_{\mathcal{N}}(\vec{x}).
$$
Here \(R_{\mathcal{N}}\ge 0\) and \(R_{\mathcal{N}}(\vec{x})=0\) iff all atoms in \(\mathcal{N}\) hold at \(\vec{x}\). As a result, Stage 1 preserves the soundness contract (non‑negative; zero only at true models) while guaranteeing that any strict decrease in the objective \textbf{necessarily reduces the distance to} \(\mathcal{L}\). Because both the projection and the evaluation are performed in floating point, Stage 1 is used to \emph{guide search} rather than to certify SAT.

The later stages supply bit‑level decision power. Stage 2 switches to an IEEE‑aligned, squared‑ULP representing function that returns zero iff all literals are satisfied as floats. Stage 3 performs a bounded (n)‑ULP refinement on the floating‑point lattice around Stage 2’s incumbent, snapping near‑miss candidates to exact models when continuous steps cannot cross the last ULP gap. The overall strategy is thus: measure \emph{true distance where we can} (the linear part), and use \emph{IEEE‑faithful objectives where we must} (ULP‑based stages), so that decreases in the objective reliably correspond to moving closer to a model.

Here we show how a naive choice breaks the principle (monotone descent) and how projection fixes it, then explain why each stage is needed for efficient IEEE‑faithful decisions.

\paragraph{From \(f_{\text{naive}}\) to \(S_1\): why projection.}

\emph{Toy constraint.} \(x == 1 \;\wedge\; y == x\). The unique model is \((1,1)\); the linear manifold is the single point \((1,1)\).

\emph{Naïve objective and its failure.} A natural first attempt is
\[
f_{\text{naive}}(x,y) \;=\; (x-1)^2 \;+\; (y-x)^2 .
\]
It measures residuals of the two equations, but it does \emph{not} align with our monotone‑descent intent: moving from \((2,2)\) to \((2,1)\) gets \emph{closer} to \((1,1)\) in Euclidean distance, yet \(f_{\text{naive}}\) \emph{increases}. Intuitively, \((y-x)^2\) penalizes motion along the manifold’s tangent direction, so the objective can go up even as the point approaches the model.

\begin{center}
\begin{tabular}{lcc}
\hline
\textbf{Point \((x,y)\)} & \textbf{Distance to \((1,1)\)} & \(\mathbf{f_{\text{naive}}(x,y)}\) \\
\hline
(2, 2) & 1.414 & 1.000 \\
(2, 1) & \textbf{1.000} & \textbf{2.000} \\
(1, 1) & 0.000 & 0.000 \\
(0, 1) & \textbf{1.000} & \textbf{2.000} \\
(0, 0) & 1.414 & 1.000 \\
\hline
\end{tabular}
\end{center}

\emph{Projection: make the proxy be the distance itself.} On \(\mathcal{L}\), we measure \emph{distance to the feasible set} rather than residuals to its defining equations. Concretely for this toy, the orthogonal projection of any \((x,y)\) onto the linear manifold is \((1,1)\). The Stage 1 objective becomes
\[
\boxed{\,S_1(x,y) \;=\; (x-1)^2 \;+\; (y-1)^2.\,}
\]
Because \(S_1\) is (by construction) the squared distance to the linear solution set, any \emph{strict decrease} in \(S_1\) necessarily means the assignment moved \emph{closer} to that set. This eliminates the tangential artifact in \(f_{\text{naive}}\) and restores predictable descent.

\emph{Why S1 does not declare SAT.} S1 is a \emph{search} stage. The projection and the distance are computed in floating point, so a tiny near‑zero can be a rounding artifact. Even when \(S_1\) is numerically small, we \emph{do not} certify SAT at S1. Instead, S1 supplies a strong initializer for the bit‑precise stages:
\begin{itemize}
  \item \emph{S2} uses pairwise \(\mathrm{ULP}^2\) to decide SAT \emph{exactly} when the objective reaches zero.
  \item \emph{S3} performs a tiny \(n\)-ULP lattice refinement around S2’s incumbent, snapping near‑misses to exact equality (SAT), else producing \emph{unsat‑guess} or \emph{timeout}.
\end{itemize}

\paragraph{Stage 2 (S2): Squared‑ULP objective — a bit‑precise representing function.}

After S1 brings us near feasibility, Stage 2 switches to a \emph{bit‑level} view aligned with IEEE‑754. Intuitively, \(\mathrm{ULP}(a,b)\) counts how many representable FP values (“steps” on the lattice) separate \(a\) and \(b\); it is 0 iff the floats are exactly equal.

\begin{itemize}
  \item \emph{Equality literal} \(f(\vec{x}) \mathtt{==} g(\vec{x})\): distance is \(\mathrm{ULP}(f(\vec{x}),g(\vec{x}))\).
  \item \emph{Inequality literal} \(f(\vec{x}) \le g(\vec{x})\): use 0 when the inequality holds; otherwise, the minimal ULP steps needed to make it true (with the standard ±1 correction for strictness).
  \item \emph{Clause} (disjunction): multiply the \emph{squared} ULP distances of its literals so that any satisfied literal zeros the product; the full S2 objective \emph{sums products over clauses}.
\end{itemize}

For the toy (each equality is its own clause), the \emph{concrete} Stage 2 objective is
\[
\boxed{\,S_2(x,y) \;=\; \mathrm{ULP}(x,1)^2 \;+\; \mathrm{ULP}(y,x)^2.\,}
\]
By construction, \(S_2(\vec{x}) \ge 0\) and \(S_2(\vec{x})=0\) \emph{iff} all constraints hold \emph{exactly} in IEEE‑754; i.e., \(S_2\) is a \emph{representing function} for \(\mathcal{C}\). Consequently, \emph{S2 can declare SAT} when its minimum reaches \(0\).

\emph{Near‑miss \(\to\) exact.} Suppose S1 yields \((x,y)\) with \(x\) one ULP above \(1\) and \(y\) two ULPs above \(x\). Then \(S_2 = 1^2 + 2^2 = 5>0\). S2 continues to adjust \((x,y)\) until both gaps are zero; on this toy it typically reaches \((1,1)\), achieving \(S_2(1,1)=0\) and reporting \emph{SAT}.

\paragraph{Stage 3 (S3): \(n\)-ULP refinement — discrete search for the last bit.}

Even with ULP‑aware optimization, a continuous method can stall a \emph{few ULPs} from equality—especially near subnormals or in tight chains of equalities—because continuous steps cannot “snap” across the final discrete gap. \emph{Stage 3} resolves this via a tiny \emph{discrete search} on the FP lattice around S2’s incumbent \((x_2^*,y_2^*)\).

Let \(\mathrm{nULP}(k,z)\) move a float \(z\) by \(k\) ULPs (positive: \texttt{nextUp}; negative: \texttt{nextDown}). We evaluate the same ULP penalties at stepped points and minimize over tiny integer offsets \(m\) (for \(x\)) and \(n\) (for \(y\)):
\[
\boxed{\,S_3(m,n) \;=\; \mathrm{ULP}\!\big(\mathrm{nULP}(m,x_2^*),\,1\big)^2
\;+\; \mathrm{ULP}\!\big(\mathrm{nULP}(n,y_2^*),\,\mathrm{nULP}(m,x_2^*)\big)^2.\,}
\]
If the best value over a small bounded neighborhood hits \(0\), S3 \emph{declares SAT} with the corresponding stepped assignment; if the minimum remains \emph{strictly positive}, we return \emph{unsat‑guess} (not a proof but empirically highly accurate); if the budget elapses first, \emph{timeout}.

\emph{Why S3 is necessary (intuition).} In a 20‑variable chain \(x_2{=}x_1, x_3{=}x_2,\ldots,x_{20}{=}x_{19}\) anchored at a tiny \(x_1{=}c\), S1 pulls toward the affine manifold and S2 aligns with bit‑level equality, but one or two coordinates can remain a couple of ULPs off—continuous updates disturb neighbors or overshoot at these scales. S3’s bounded lattice search (e.g., steps in \(\{-1,0,1\}\)) jointly tweaks last bits and snaps the whole vector to the exact lattice point, driving the objective to \(0\) when a model exists.

\paragraph{Takeaway.}
Replacing \(f_{\text{naive}}\) with the \emph{projection‑based distance} \(S_1\) restores the intended "closer $\implies$ smaller" geometry on \(\mathcal{L}\) and yields a robust initializer. \emph{S2} contributes an \emph{IEEE‑faithful} \(\mathrm{ULP}^2\) \emph{representing function} that certifies SAT when it reaches zero. \emph{S3} adds a small \(n\)-ULP \emph{lattice refinement} to bridge the last‑bit gap when smooth optimization stalls. Together, the pipeline \(S_1 \rightarrow S_2 \rightarrow S_3\) provides predictable progress, exact satisfiability decisions when possible, and principled \emph{unsat‑guess} or \emph{timeout} otherwise.

\section{Technical Approach and Theory}
\label{sec:theory}

This section makes precise how \textsc{StageSAT} works and why it is effective. We formalize the three stages and establish three facts:
(i) \textbf{Stage~1} satisfies a \textbf{partial monotone descent} guarantee on the \textbf{linear} part;
(ii) \textbf{Stage~2} and \textbf{Stage~3} are \textbf{representing functions} for the constraint set; and
(iii) the overall procedure is \textbf{sound} (no false SAT), while necessarily \textbf{incomplete}.
These results explain \textsc{StageSAT}'s \textbf{accuracy} and \textbf{scalability} in practice.

\subsection{Notation for QF\_FP setting}

\paragraph{Language.}
We consider quantifier‑free floating‑point (IEEE‑754) formulas. Terms are built from FP variables and constants using standard IEEE‑754 arithmetic; atoms have the form \(e_1 \ \mathsf{op}\ e_2\) with \(\mathsf{op}\in\{\mathtt{==},\le,<,\ge,>\}\).
Mixed precision (FP32/FP64) is allowed; each atom is evaluated in its declared format.

\paragraph{CNF and models.}
Let \(\mathcal{C}\) be a \emph{CNF}: a conjunction of clauses \(\phi\) (each clause is a disjunction of literals).
We write \(\vec{x}\models\mathcal{C}\) if all clauses evaluate to true under IEEE‑754 at \(\vec{x}\).
The \emph{model set} is \(\mathcal{M}_{\mathcal{C}}:=\{\vec{x}:\vec{x}\models\mathcal{C}\}\).
For brevity we use \(\sum_{\phi\in\mathcal{C}}(\cdot)\) (or simply \(\sum_{\phi}(\cdot)\)) to denote summation over clauses.

\paragraph{Euclidean distance to a set.}
For \(S\subseteq\mathbb{R}^d\), the (Euclidean) distance is
\[
d(\vec{x},S)\ :=\ \inf_{\vec{y}\in S}\ \|\vec{x}-\vec{y}\|_2.
\]
When \(S\) is an affine set (e.g., a linear manifold), the infimum is attained at the \emph{orthogonal projection} of \(\vec{x}\) onto \(S\).

\medskip
\noindent\emph{We keep formal objective definitions local to each stage to avoid redundancy and to highlight their distinct roles.}

\subsection{Stage~1 (S1): projection‑aided square objective and partial monotone descent}

Let \(\mathcal{L}=\{\vec{x}:A\vec{x}=\vec{b}\}\) denote the \emph{linear equalities} extracted from \(\mathcal{C}\), and let \(\mathcal{N}\) be the remaining atoms (nonlinear equalities and all inequalities).

\paragraph{Square distance for literals and clauses (multiplication in \(\lor\)).}
For a literal \(\ell\):
\begin{itemize}
  \item equality \(h(\vec{x})\ \mathtt{==}\ 0\): \(\mathrm{dist}_{\square}(\ell;\vec{x}) := h(\vec{x})^2\);
  \item inequality \(f(\vec{x})\le g(\vec{x})\): \(\mathrm{dist}_{\square}(\ell;\vec{x}) := \max\{0,\,f(\vec{x})-g(\vec{x})\}^2\).
\end{itemize}
For a \emph{clause} \(\phi\) (a disjunction of literals), we encode \(\lor\) by \emph{multiplication}:
\[
\mathrm{Dist}_{\square}(\phi;\vec{x})
\ :=\
\prod_{\ell\in\phi}\ \mathrm{dist}_{\square}(\ell;\vec{x}),
\]
so if \emph{any} literal in \(\phi\) is satisfied, one factor is \(0\) and the clause distance is \(0\) (as in XSat).
Let \(\mathcal{C}_{\mathcal{N}}\) be the set of clauses formed only from atoms in \(\mathcal{N}\). The \emph{S1 objective} is
\[
S_1(\vec{x})
\ =\
\underbrace{\|\vec{x}-\Pi_{\mathcal{L}}(\vec{x})\|_2^2}_{\text{exact squared distance to } \mathcal{L}}
\;+\;
\sum_{\phi\in\mathcal{C}_{\mathcal{N}}}\ \mathrm{Dist}_{\square}(\phi;\vec{x}).
\]

\paragraph{Projection (definition and closed form).}
\[
\Pi_{\mathcal{L}}(\vec{x})\ =\ \vec{x} - A^\top\!\big(AA^\top\big)^{\dagger}\!(A\vec{x}-\vec{b}),
\qquad
\|\vec{x}-\Pi_{\mathcal{L}}(\vec{x})\|_2^2
\ =\ \Big\|A^\top\!\big(AA^\top\big)^{\dagger}\!(A\vec{x}-\vec{b})\Big\|_2^2,
\]
with \((\cdot)^{\dagger}\) the Moore–Penrose pseudoinverse\cite{Penrose_1955}.

\paragraph{Monotone Descent Property (MDP)}
For an objective \(R\), we say \emph{\(R\) satisfies monotone descent for \(\mathcal{C}\)} if
\[
\boxed{\quad
R(\vec{x}')\ <\ R(\vec{x})
\ \Longrightarrow\
d(\vec{x}',\,\mathcal{M}_{\mathcal{C}})\ <\ d(\vec{x},\,\mathcal{M}_{\mathcal{C}}).
\quad}
\]

\paragraph{\textbf{Lemma 1 (projection facts).}}
\(\Pi_{\mathcal{L}}(\vec{x})\) is the unique point in \(\mathcal{L}\) closest to \(\vec{x}\), and \(\|\vec{x}-\Pi_{\mathcal{L}}(\vec{x})\|_2=d(\vec{x},\mathcal{L})\).
The closed forms above hold.

\paragraph{\textbf{Lemma 2 (partial MDP for S1 when \(\mathcal{N}\) holds).}}
If \(\vec{x}\models \mathcal{N}\) and \(\vec{x}'\models \mathcal{N}\), then
\[
S_1(\vec{x}')\ <\ S_1(\vec{x})
\ \Longleftrightarrow\
\|\vec{x}'-\Pi_{\mathcal{L}}(\vec{x}')\|_2^2\ <\ \|\vec{x}-\Pi_{\mathcal{L}}(\vec{x})\|_2^2,
\]
and therefore \(d(\vec{x}',\mathcal{M}_{\mathcal{C}})<d(\vec{x},\mathcal{M}_{\mathcal{C}})\), since under \(\mathcal{N}\) the (local) model set coincides with \(\mathcal{L}\).
\emph{Interpretation.} When the nonlinear/inequality part already holds, S1’s decrease exactly tracks movement toward the model set.
S1 provides a strong \emph{starting point}; it is \emph{not} used to declare SAT.

\subsection{Stage~2 (S2): squared‑ULP objective is a representing function}

We now align the objective with IEEE‑754 semantics, again respecting CNF by multiplying inside each disjunction.

\paragraph{ULP distance for literals and clauses (multiplication in \(\lor\)).}
For a literal \(\ell\):
\begin{itemize}
  \item equality \(f(\vec{x})\ \mathtt{==}\ g(\vec{x})\): \(d_{\mathrm{ulp}}(\ell;\vec{x})\) is the number of adjacent FP steps (via IEEE \texttt{nextUp}/\texttt{nextDown}) between the IEEE values of \(f(\vec{x})\) and \(g(\vec{x})\) (0 iff bit‑equal);
  \item inequality \(f(\vec{x})\le g(\vec{x})\): \(d_{\mathrm{ulp}}(\ell;\vec{x})=0\) if it holds; otherwise the \emph{minimal} FP steps needed to make it hold (used purely as a distance).
\end{itemize}
For a \emph{clause} \(\phi\), define
\[
\mathrm{Dist}_{\mathrm{ulp}}(\phi;\vec{x})
\ :=\
\prod_{\ell\in\phi}\ d_{\mathrm{ulp}}(\ell;\vec{x})^2,
\]
so a satisfied literal forces the product to \(0\).

\paragraph{Representing function and S2 definition.}
A scalar \(R\) is a \emph{representing function} for \(\mathcal{C}\) if
\[
R(\vec{x})\ \ge\ 0
\quad\text{and}\quad
R(\vec{x})=0\ \Longleftrightarrow\ \vec{x}\models\mathcal{C}.
\]
We set
\[
S_2(\vec{x})\ :=\ \sum_{\phi\in\mathcal{C}}\ \mathrm{Dist}_{\mathrm{ulp}}(\phi;\vec{x}).
\]

\paragraph{\textbf{Lemma 3 (S2 is representing).}}
\(S_2\) is a representing function for \(\mathcal{C}\).
\emph{Reason.} Each clause distance is non‑negative and is \(0\) iff the clause is satisfied; the sum is \(0\) iff all clauses are satisfied, i.e., \(\vec{x}\models\mathcal{C}\).

\subsection{Stage~3 (S3): n‑ULP refinement is a representing function (discrete steps)}

Let \(\vec{x}_2^{\star}\) be the \emph{current best assignment returned by S2}. For an integer vector \(\vec{n}\) (signed ULP steps applied component‑wise via IEEE \texttt{nextUp}/\texttt{nextDown}), define the clause distance at the \emph{stepped} point and sum over clauses:
\[
\mathrm{Dist}_{\mathrm{ulp}}\!\big(\phi;\ \mathrm{nULP}(\vec{n},\vec{x}_2^{\star})\big)
\ :=\
\prod_{\ell\in\phi}\ d_{\mathrm{ulp}}\!\big(\ell;\ \mathrm{nULP}(\vec{n},\vec{x}_2^{\star})\big)^2,
\]
\[
S_3(\vec{n})\ :=\ \sum_{\phi\in\mathcal{C}}\ \mathrm{Dist}_{\mathrm{ulp}}\!\big(\phi;\ \mathrm{nULP}(\vec{n},\vec{x}_2^{\star})\big).
\]

\paragraph{\textbf{Lemma 4 (S3 is representing at stepped points).}}
For all \(\vec{n}\), \(S_3(\vec{n})\ge 0\) and \(S_3(\vec{n})=0\ \Longleftrightarrow\ \mathrm{nULP}(\vec{n},\vec{x}_2^{\star})\models\mathcal{C}\).
\emph{Interpretation.} S3 searches the FP lattice around the S2 result and can \emph{snap} a near‑miss to an exact model while preserving CNF semantics (product for \(\lor\), sum for \(\land\)).

\subsection{Overall procedure and guarantees}
\label{sec:overall-procedure}

\paragraph{Inputs.}
\begin{itemize}
  \item A CNF constraint $\mathcal{C}$ over \textsc{QF\_FP} atoms.
  \item A black-box $\textsc{MP\_Inverse}(M)$ that returns the Moore--Penrose pseudoinverse $M^{\dagger}$.
  \item A black-box global minimizer $\textsc{Global\_Min}(F,\mathrm{init},\mathrm{budget}) \to (\hat{\vec{x}}, F(\hat{\vec{x}}), \texttt{status})$.
  \item Stage budgets $\mathrm{budget}_1,\mathrm{budget}_2,\mathrm{budget}_3$.
\end{itemize}

\paragraph{Algorithm.}
\begin{enumerate}
  \item  \textbf{ Build the objective functions from $\mathcal{C}$.}
  Collect all linear equalities into a system $A\vec{x} == \vec{b}$ and let the remaining atoms form $\mathcal{N}$. 
  
  \begin{enumerate}
    \item[]   Define $S_1(\vec{x})$ as the sum of: (i) the squared distance from $\vec{x}$ to $\mathcal{L}$ via $\Pi_{\mathcal{L}}$ and (ii) the clause-level product of squared violation terms for $\mathcal{C}_N$ (cf.\ Section~\ref{sec:theory}).\;
    \item[]  Define $S_2(\vec{x})$ as the sum over clauses of the product of squared per-literal ULP distances (IEEE-aligned representing function; cf.\ Section~\ref{sec:theory}).\;
    \item[] Given an incumbent $\vec{x}_2^{\star}$ from $S_2$, define $S_3(\vec{n};\vec{x}_2^{\star})$ as the $S_2$ objective evaluated at the FP-lattice point $\mathrm{nULP}(\vec{n},\vec{x}_2^{\star})$ obtained by stepping each coordinate by $n_j$ ULPs (cf.\ Section~\ref{sec:theory}).\;
  \end{enumerate}

  \item \textbf{Minimize S1 (fast descent).}
  $(\vec{x}_1^{\star}, v_1, \_) \leftarrow \textsc{Global\_Min}\!\big(S_1,\ \textsc{MultiStartBox}(\mathcal{C}),\ \mathrm{budget}_1\big)$.
  \emph{Note:} we do \emph{not} decide SAT here.

  \item \textbf{Minimize S2 (bit-precise optimization).}
  $(\vec{x}_2^{\star}, v_2, \_) \leftarrow \textsc{Global\_Min}\!\big(S_2,\ \vec{x}_1^{\star},\ \mathrm{budget}_2\big)$.
  If $v_2 = 0$, return \textbf{sat} with $\vec{x}_2^{\star}$.

  \item \textbf{Minimize S3 (n-ULP refinement over the FP lattice).}
  Choose per-dimension integer step bounds $\vec{N}$.
  Run a bounded discrete search on $\vec{n}\in[-\vec{N},\vec{N}]\cap\mathbb{Z}^d$ to minimize $S_3(\vec{n};\vec{x}_2^{\star})$.
  If the best value $v_3 = 0$, return \textbf{sat} with $\mathrm{nULP}(\vec{n}^{\star},\vec{x}_2^{\star})$.

  \item Otherwise, return \textbf{unsat-guess} with score $\min(v_2, v_3)$, or \textbf{timeout} 
  if the time budget expires. 
\end{enumerate}

\begin{lemma}[Outcome trichotomy]
\label{thm:trichotomy}
For any CNF \(\mathcal{C}\) over IEEE-754 atoms and any finite time budget, the procedure above returns exactly one of \(\{\textbf{sat},\ \textbf{unsat-guess},\ \textbf{timeout}\}\).
\end{lemma}

\begin{proof}
By case analysis on the control flow. Steps (3)--(4) return \textbf{sat} upon reaching an objective value \(0\) with validation; otherwise, upon termination with a strictly positive best value, step (5) returns \textbf{unsat-guess}; if the time budget elapses before either condition, step (4) returns \textbf{timeout}. These cases are mutually exclusive and exhaustive.
\end{proof}

\begin{theorem}[SAT soundness]
\label{thm:sat-soundness}
If the procedure returns \textbf{sat} with assignment \(\hat{\vec{x}}\), then \(\hat{\vec{x}}\models\mathcal{C}\).
\end{theorem}

\begin{proof}
A \textbf{sat} result arises only in S2 or S3. In S2, \(S_2(\hat{\vec{x}})=0\) and Lemma~3 (S2 is representing) imply \(\hat{\vec{x}}\models\mathcal{C}\). In S3, \(S_3(\vec{n}^{\star})=0\) and Lemma~4 (S3 is representing at stepped points) imply \(\mathrm{nULP}(\vec{n}^{\star},\vec{x}_2^{\star})\models\mathcal{C}\), which is exactly the returned \(\hat{\vec{x}}\). In both cases, we additionally validate under IEEE-754 before returning. Hence any \textbf{sat} output is a genuine model of \(\mathcal{C}\).
\end{proof}

\begin{corollary}[No false SAT]
\label{cor:no-false-sat}
The procedure never reports \textbf{sat} on a non-model.
\end{corollary}

\begin{proof}
Immediate from Theorem~\ref{thm:sat-soundness}.
\end{proof}

\begin{proposition}[Incompleteness]
\label{prop:incompleteness}
There exist satisfiable \(\mathcal{C}\) and finite budgets for which the procedure returns \textbf{unsat-guess} or \textbf{timeout}.
\end{proposition}

\begin{proof}[Proof sketch]
Numeric minimization in S1/S2 may converge to non-global minima; S3 searches only a bounded set of integer step vectors. With finite time, neither guarantees reaching a zero even when one exists.
\end{proof}

\paragraph{Discussion.}
Although incomplete, the design provides strong practical behavior: (1) S1's projection yields predictable progress on linear equalities, stabilizing the search on large instances; (2) S2 and S3 are IEEE-aligned representing functions, ensuring no false SAT and making last-bit gaps explicit; and (3) S3's discrete refinement over the FP lattice snaps near-misses to exact models without brute force. These properties explain the accuracy and scalability observed in our evaluation.

\section{Implementation Details}
\label{sec:implementation}

\emph{Projection for $S_1$.} We assemble $A$ and $\vec{b}$ from the FP-linear equalities detected in $\mathcal{C}$ and compute the projection using $A A^{\top}$ and a rank-agnostic Moore--Penrose pseudoinverse. When $A A^{\top}$ is numerically singular or when no linear equalities exist, we gracefully skip the projection term and keep the squared penalties for the remaining atoms; the objective remains well-defined.

\emph{Clause aggregation.} Inside each disjunction we multiply per-literal distances, and across clauses we sum. This “product-in-$\lor$ / sum-over-clauses” scheme is used consistently in both $S_1$ (with squared distances) and $S_2$/$S_3$ (with squared ULP distances), matching the representing-function contracts from Section~\ref{sec:theory}.

\emph{Global minimization.} \textsc{Global\_Min} is instantiated by a multi-start strategy with a derivative-free local solver; we terminate early on zero objectives and parallelize starts across cores. Stage budgets $\mathrm{budget}_1,\mathrm{budget}_2,\mathrm{budget}_3$ are configurable, with $S_1$ typically receiving the largest share to shape a good basin, and $S_3$ using a small, bounded neighborhood.

\emph{ULP distances and mixed precision.} We implement per-literal ULP distances in IEEE~754, handling strict inequalities with the standard $\pm 1$ offset and dispatching to FP32/FP64 variants per variable. Mixed-precision atoms are supported by per-operand rules consistent with Section~\ref{sec:theory}.

\emph{Constant-time nULP stepping.} The primitive $\mathrm{nULP}(\vec{n},\vec{x})$ is $O(1)$ per coordinate: we reinterpret each FP value as a monotone signed integer index over the FP lattice, add $n_j$ as a \emph{single} integer operation, and bit-cast back to FP, clamping to the nearest finite value when necessary. Subnormals and signed zeros are preserved by construction.

\emph{Numerical hygiene.} We exclude NaN/$\pm\infty$ from the search space, precompute per-variable sorts (FP32/FP64) to dispatch kernels efficiently, and cache matrix factorizations for the projection term when $A$ is unchanged across restarts.

% \vspace{0.25em}
\noindent
Only $S_2$ and $S_3$ can return \textbf{sat}; $S_1$ is used solely for fast initialization. This mirrors the theory in Section~\ref{sec:theory}: $S_1$ provides partial monotone descent on the linear part; $S_2$ and $S_3$ are representing functions (zero if and only if $\vec{x}\models\mathcal{C}$); and the overall procedure is sound (no false \textbf{sat}), incomplete by necessity, and empirically accurate and scalable.

Appendix \ref{sec:pseudocode} describes the implemented algorithms. Our  supplementary artifact (anonymized) packages the StageSat tool, a usage guide, benchmarks, and experiment scripts, allowing reviewers to reproduce the tables and figures reported in this paper.

\section{Experiments}
\label{sec:Experiments}
\subsection{Executive Summary}

StageSAT demonstrates robust performance and reliability on widely-used floating-point satisfiability benchmarks:
\begin{itemize}
    \item 

\emph{Broad coverage}: StageSAT produced results on 45 of 49 formulas in the hardest FP benchmark suite (versus 39 by the best solver), and handled more benchmarks overall than any competing solver thanks to a much lower timeout rate.

    \item 
\emph{Near-perfect accuracy in SAT detection}: StageSAT found solutions for all but two satisfiable formulas in our benchmark suite (~99.4\% recall) and never reported a satisfiable result for any unsatisfiable formula (0\% false SAT). These outcomes perfectly aligned with the complete solver’s UNSAT verdicts, with no incorrect models produced in any case.

    \item 
\emph{Significant time improvement}: StageSAT delivered these results much faster than the complete solver, achieving an overall solving speedup of 5--10$\times$ faster on average. It also encountered far fewer timeouts, solving nearly all instances within the time limit in our benchmarks.

    \item 
\emph{Essential pipeline}: Removing any one of StageSAT’s three optimization stages significantly reduced the number of instances it could handle (or led to missed solutions), confirming that every component of its design is critical for effectiveness.

\end{itemize}

% In our experiments, StageSAT achieved the broadest coverage of problems and faster solving times than both \emph{incomplete} and \emph{complete} baseline solvers. For example, on the hardest MathSAT-Large suite (49 formulas), StageSAT produced results for 45 instances within a 20-minute timeout, slightly exceeding the best complete solver (Bitwuzla, 39 solved) and far surpassing other numerical solvers. Notably, StageSAT found valid models for \emph{almost all} satisfiable formulas (covering all but two satisfiable instances) and never returned a spurious solution – yielding nearly 100\% SAT coverage (correct SAT recall) with 0 false positives. All remaining formulas that were truly unsatisfiable were consistently reported as “unsat-guess” by StageSAT, with no contradictions observed against proven UNSAT results. An ablation study confirms that each component of StageSAT’s three-stage optimization pipeline is essential: removing or modifying any stage led to fewer solved cases or missed solutions. In summary, StageSAT offers a new, optimization-driven approach to FP solving that complements traditional SMT solvers by quickly finding models in difficult satisfiable cases while maintaining a high level of soundness on unsatisfiable instances.

\subsection{Experimental Setup}

We evaluate \emph{StageSAT} on five suites spanning a broad range of floating‑point SMT problems: \emph{MathSAT‑Small} (130 files, $\leq$10\,KB), \emph{MathSAT‑Middle} (35 files, 11–20\,KB), \emph{MathSAT‑Large} (49 files, $>$20\,KB), \emph{Grater} (118 files; 5 non‑RNE cases filtered), and \emph{JFS} (111 files).
The Large set serves as the primary stress test; the others provide breadth across sizes and sources. We compare two solver categories: \emph{incomplete} solvers—\emph{XSat}, \emph{goSAT}, \emph{Grater}, and \emph{JFS} (optimization/fuzzing based; cannot prove UNSAT)—and \emph{complete} solvers—\emph{Z3}, \emph{CVC5}, \emph{MathSAT}, and \emph{Bitwuzla} (bit‑precise decision procedures). All experiments ran on an \emph{Intel Core i9‑13900HK} (14 cores), \emph{32\,GB RAM}, \emph{Ubuntu~24.04.3~LTS}, with a \emph{20‑minute} wall‑clock cap per instance (except for a stress test which we use 48 hours and results are reported in Table \ref{tab:rq2-detailed-stagesat-bitwuzla}). Incomplete solvers (including \emph{StageSAT}) were run $5\times$ per benchmark.

Our evaluation uses the following metrics. \emph{SAT coverage (a.k.a.\ SAT recall)} is the fraction of ground‑truth SAT instances for which the solver returns \texttt{sat} with a validated model. We obtain all ground-truth SAT/UNSAT labels from complete solvers: whenever at least one complete solver finishes within the 48-hour limit, we adopt its result (and in all such cases, completed solvers agree). For one particularly hard benchmark, every complete solver hit the 48-hour timeout, so we reran the fastest complete solver, Bitwuzla, with a longer timeout and used its final result as ground truth.
\[
\mathrm{SAT\text{-}Recall}
\;=\;
\frac{\#\{\text{reported SAT } \wedge \, \text{ground-truth=SAT}\}}
     {\#\{\text{ground-truth=SAT}\}} \, .
\]

\emph{Timeout rate} is the share of all benchmarks with no verdict within 20 minutes:
\[
\mathrm{Timeout\text{-}rate}
=\frac{\#\{\text{timeout}\}}{\#\{\text{all benchmarks}\}}.
\]
\emph{Average (mean) time} is the arithmetic mean per suite with timeouts counted at the cap (to reflect robustness under the budget). \emph{Median time} is reported alongside the mean to reflect typical behavior and reduce heavy‑tail effects.

We do not use \textbf{SAT precision} (the fraction of reported \texttt{sat} verdicts that are truly SAT) for comparison. In our experiments, it is effectively 100\% for all solvers—rare exceptions occurred only when \emph{Grater} accepted sub-tolerance objectives—making it non-discriminative. Moreover, precision can be inflated by reporting \texttt{sat} only on easy cases, while SAT coverage can be inflated by labeling everything \texttt{sat}. We therefore emphasize \textbf{SAT coverage, timeout rate, and time statistics}. For \emph{StageSAT}, \texttt{unsat-guess} is a distinct, non-proof outcome that we never count as proved UNSAT or as ``solved''.

\subsection{Quantitative Results}

\begin{table}[htbp]
  \centering
  \caption{Comparison on MathSAT-Large benchmarks: \textsc{StageSAT} vs.\ incomplete solvers. Timeout as 20 mins.}
  \label{tab:rq1-detailed}
  \resizebox{\textwidth}{!}{%
  \small
  \setlength{\tabcolsep}{3pt}
  \begin{tabular}{@{}l r r | cc | cc | cc | cc | cc@{}}
    \toprule
    \multirow{2}{*}{Benchmark} & \multirow{2}{*}{Size(byte)} & \multirow{2}{*}{\#Vars} &
    \multicolumn{2}{c|}{\textsc{XSat}} &
    \multicolumn{2}{c|}{\textsc{goSAT}} &
    \multicolumn{2}{c|}{\textsc{Grater}} &
    \multicolumn{2}{c|}{\textsc{JFS}} &
    \multicolumn{2}{c}{\textsc{StageSAT}} \\
    \cmidrule(lr){4-5} \cmidrule(lr){6-7} \cmidrule(lr){8-9} \cmidrule(lr){10-11} \cmidrule(l){12-13}
    & & & Verdict & Time(s) & Verdict & Time(s) & Verdict & Time(s) & Verdict & Time(s) & Verdict & Time(s) \\
    \midrule
    {\footnotesize\texttt{sqrt.c.10}}                      & 21701 & 26  & unsat & 7.81 & unknown & 0.72 & timeout & >1200 & timeout & >1200 & sat & 102.48 \\
    {\footnotesize\texttt{test\_v5\_r15\_vr5\_c1\_s8246}}  & 21791 & 5   & unsat & 1.62 & unknown & 0.02 & timeout & >1200 & timeout & >1200 & unsat-guess & 10.76 \\
    {\footnotesize\texttt{test\_v5\_r15\_vr1\_c1\_s26845}} & 21811 & 5   & unsat & 1.24 & unknown & 0.02 & timeout & >1200 & timeout & >1200 & unsat-guess & 7.00 \\
    {\footnotesize\texttt{test\_v5\_r15\_vr10\_c1\_s25268}}& 21818 & 5   & unsat & 1.54 & unknown & 0.02 & timeout & >1200 & timeout & >1200 & unsat-guess & 8.87 \\
    {\footnotesize\texttt{test\_v5\_r15\_vr5\_c1\_s26657}} & 22070 & 5   & unsat & 1.48 & unknown & 0.02 & timeout & >1200 & timeout & >1200 & unsat-guess & 6.70 \\
    {\footnotesize\texttt{test\_v5\_r15\_vr1\_c1\_s32559}} & 22072 & 5   & unsat & 1.58 & unknown & 0.02 & timeout & >1200 & timeout & >1200 & unsat-guess & 5.78 \\
    {\footnotesize\texttt{test\_v5\_r15\_vr1\_c1\_s8236}}  & 22072 & 5   & unsat & 1.22 & unknown & 0.02 & timeout & >1200 & timeout & >1200 & unsat-guess & 5.88 \\
    {\footnotesize\texttt{test\_v5\_r15\_vr5\_c1\_s23844}} & 22072 & 5   & unsat & 1.66 & unknown & 0.02 & timeout & >1200 & timeout & >1200 & unsat-guess & 9.21 \\
    {\footnotesize\texttt{test\_v5\_r15\_vr10\_c1\_s14516}}& 22252 & 5   & unsat & 1.40 & unknown & 0.02 & timeout & >1200 & timeout & >1200 & unsat-guess & 10.56 \\
    {\footnotesize\texttt{qurt.c.5}}                       & 23169 & 30  & unsat & 9.44 & unknown & 0.83 & sat & 1177.90 & timeout & >1200 & unsat-guess & 156.14 \\
    {\footnotesize\texttt{test\_v7\_r12\_vr5\_c1\_s29826}} & 23736 & 7   & sat & 0.18 & sat & 0.02 & sat & 0.07 & sat & 3.25 & sat & 0.06 \\
    {\footnotesize\texttt{test\_v7\_r12\_vr10\_c1\_s15994}}& 23828 & 7   & sat & 0.14 & sat & 0.02 & sat & 0.21 & sat & 16.94 & sat & 0.06 \\
    {\footnotesize\texttt{test\_v7\_r12\_vr10\_c1\_s30410}}& 24070 & 7   & sat & 1.40 & sat & 0.02 & sat & 0.89 & timeout & >1200 & sat & 0.09 \\
    {\footnotesize\texttt{test\_v7\_r12\_vr5\_c1\_s14336}} & 24250 & 7   & sat & 0.14 & sat & 0.03 & sat & 0.07 & sat & 0.38 & sat & 0.05 \\
    {\footnotesize\texttt{test\_v7\_r12\_vr5\_c1\_s8938}}  & 24251 & 7   & sat & 0.14 & sat & 0.02 & sat & 0.07 & sat & 0.14 & sat & 0.05 \\
    {\footnotesize\texttt{test\_v7\_r12\_vr1\_c1\_s10576}} & 24274 & 7   & unsat & 3.04 & unknown & 0.03 & timeout & >1200 & timeout & >1200 & unsat-guess & 14.56 \\
    {\footnotesize\texttt{test\_v7\_r12\_vr1\_c1\_s22787}} & 24345 & 7   & unsat & 2.24 & unknown & 0.04 & timeout & >1200 & timeout & >1200 & unsat-guess & 11.65 \\
    {\footnotesize\texttt{test\_v7\_r12\_vr10\_c1\_s18160}}& 24437 & 7   & unsat & 2.40 & unknown & 0.04 & timeout & >1200 & timeout & >1200 & unsat-guess & 21.18 \\
    {\footnotesize\texttt{test\_v7\_r12\_vr1\_c1\_s703}}   & 24441 & 7   & unsat & 2.76 & unknown & 0.03 & timeout & >1200 & timeout & >1200 & unsat-guess & 16.54 \\
    {\footnotesize\texttt{sin2.c.15}}                      & 25235 & 52  & sat & 25.81 & unknown & 1.78 & sat & 13.88 & timeout & >1200 & sat & 3.92 \\
    {\footnotesize\texttt{gaussian.c.25}}                  & 29883 & 79  & sat & 0.72 & unknown & 2.52 & sat & 9.12 & timeout & >1200 & sat & 3.38 \\
    {\footnotesize\texttt{sqrt.c.15}}                      & 32192 & 36  & unsat & 13.10 & unknown & 1.06 & timeout & >1200 & timeout & >1200 & sat & 101.76 \\
    {\footnotesize\texttt{test\_v7\_r17\_vr5\_c1\_s2807}}  & 32711 & 7   & unsat & 2.20 & unknown & 0.04 & timeout & >1200 & timeout & >1200 & unsat-guess & 15.21 \\
    {\footnotesize\texttt{test\_v7\_r17\_vr1\_c1\_s30331}} & 32876 & 7   & unsat & 2.58 & unknown & 0.05 & timeout & >1200 & timeout & >1200 & unsat-guess & 14.96 \\
    {\footnotesize\texttt{test\_v7\_r17\_vr5\_c1\_s25451}} & 32964 & 7   & unsat & 2.30 & unknown & 0.04 & timeout & >1200 & timeout & >1200 & unsat-guess & 11.58 \\
    {\footnotesize\texttt{sin2.c.20}}                      & 33016 & 67  & unsat & 41.81 & unknown & 1.74 & sat & 33.89 & timeout & >1200 & sat & 169.92 \\
    {\footnotesize\texttt{test\_v7\_r17\_vr10\_c1\_s8773}} & 33151 & 7   & sat & 0.74 & sat & 0.03 & sat & 0.35 & timeout & >1200 & sat & 0.06 \\
    {\footnotesize\texttt{test\_v7\_r17\_vr5\_c1\_s4772}}  & 33222 & 7   & unsat & 2.30 & unknown & 0.05 & timeout & >1200 & timeout & >1200 & unsat-guess & 35.20 \\
    {\footnotesize\texttt{test\_v7\_r17\_vr1\_c1\_s23882}} & 33226 & 7   & sat & 0.14 & sat & 0.03 & sat & 0.29 & timeout & >1200 & sat & 0.05 \\
    {\footnotesize\texttt{test\_v7\_r17\_vr1\_c1\_s24331}} & 33226 & 7   & unsat & 3.12 & unknown & 0.04 & timeout & >1200 & timeout & >1200 & unsat-guess & 11.63 \\
    {\footnotesize\texttt{test\_v7\_r17\_vr10\_c1\_s3680}} & 33335 & 7   & unsat & 3.06 & unknown & 0.05 & timeout & >1200 & timeout & >1200 & unsat-guess & 13.26 \\
    {\footnotesize\texttt{test\_v7\_r17\_vr10\_c1\_s18654}}& 33410 & 7   & sat & 0.16 & sat & 0.03 & sat & 5.73 & timeout & >1200 & sat & 0.05 \\
    {\footnotesize\texttt{sin.c.25}}                       & 40536 & 81  & unsat & 76.20 & unknown & 2.31 & sat & 98.30 & timeout & >1200 & sat & 294.71 \\
    {\footnotesize\texttt{sin2.c.25}}                      & 40747 & 82  & unsat & 61.49 & unknown & 2.84 & sat & 93.33 & timeout & >1200 & sat & 8.62 \\
    {\footnotesize\texttt{sqrt.c.20}}                      & 46804 & 63  & unsat & 45.91 & unknown & 1.73 & timeout & >1200 & timeout & >1200 & sat & 498.60 \\
    {\footnotesize\texttt{sqrt.c.25}}                      & 46804 & 63  & unsat & 46.33 & unknown & 1.82 & timeout & >1200 & timeout & >1200 & sat & 581.95 \\
    {\footnotesize\texttt{qurt.c.10}}                      & 47946 & 60  & unsat & 22.89 & unknown & 1.28 & timeout & >1200 & timeout & >1200 & unsat-guess & 393.02 \\
    {\footnotesize\texttt{qurt.c.15}}                      & 73125 & 90  & unsat & 53.40 & unknown & 2.95 & timeout & >1200 & timeout & >1200 & unsat-guess & 728.38 \\
    {\footnotesize\texttt{gaussian.c.75}}                  & 89686 & 229 & sat & 20.05 & unknown & 13.31 & sat & 542.51 & timeout & >1200 & sat & 36.45 \\
    {\footnotesize\texttt{qurt.c.20}}                      & 93126 & 114 & unsat & 68.26 & unknown & 3.96 & timeout & >1200 & timeout & >1200 & timeout & >1200 \\
    {\footnotesize\texttt{qurt.c.25}}                      & 93126 & 114 & unsat & 76.92 & unknown & 4.31 & timeout & >1200 & timeout & >1200 & timeout & >1200 \\
    {\footnotesize\texttt{sin2.c.75}}                      & 119790 & 231 & unsat & 242.70 & unknown & 11.70 & timeout & >1200 & timeout & >1200 & sat & 363.37 \\
    {\footnotesize\texttt{sin.c.75}}                       & 119794 & 231 & unsat & 214.49 & unknown & 12.32 & timeout & >1200 & timeout & >1200 & sat & 660.70 \\
    {\footnotesize\texttt{gaussian.c.125}}                 & 150792 & 379 & sat & 13.20 & unknown & 40.24 & sat & 89.02 & timeout & >1200 & sat & 90.62 \\
    {\footnotesize\texttt{sin.c.125}}                      & 200503 & 381 & unsat & 1000.64 & unknown & 34.38 & error & -- & timeout & >1200 & sat & 379.39 \\
    {\footnotesize\texttt{sin2.c.125}}                     & 200503 & 381 & unsat & 984.14 & unknown & 33.26 & error & -- & timeout & >1200 & sat & 313.85 \\
    {\footnotesize\texttt{gaussian.c.175}}                 & 210711 & 529 & unsat & 1101.04 & unknown & 80.37 & timeout & >1200 & timeout & >1200 & sat & 273.68 \\
    {\footnotesize\texttt{sin2.c.175}}                     & 280962 & 531 & timeout & >1200 & unknown & 76.20 & error & -- & timeout & >1200 & timeout & >1200 \\
    {\footnotesize\texttt{sin.c.175}}                      & 280984 & 531 & timeout & >1200 & unknown & 72.63 & error & -- & timeout & >1200 & timeout & >1200 \\
    \midrule
    \multicolumn{3}{l|}{\textbf{SAT Coverage}} & 
    \multicolumn{2}{c|}{\textbf{46.2\%}} & 
    \multicolumn{2}{c|}{\textbf{30.8\%}} & 
    \multicolumn{2}{c|}{\textbf{57.7\%}} & 
    \multicolumn{2}{c|}{\textbf{15.4\%}} & 
    \multicolumn{2}{c}{\textbf{92.3\%}} \\
    \multicolumn{3}{l|}{\textbf{Timeout Rate}} & 
    \multicolumn{2}{c|}{\textbf{4.1\%}} & 
    \multicolumn{2}{c|}{\textbf{0.0\%}} & 
    \multicolumn{2}{c|}{\textbf{51.0\%}} & 
    \multicolumn{2}{c|}{\textbf{91.8\%}} & 
    \multicolumn{2}{c}{\textbf{8.2\%}} \\
    \multicolumn{3}{l|}{\textbf{Average Time(s)}} & 
    \multicolumn{2}{c|}{\textbf{134.02}} & 
    \multicolumn{2}{c|}{\textbf{8.27}} & 
    \multicolumn{2}{c|}{\textbf{819.24}} & 
    \multicolumn{2}{c|}{\textbf{1102.46}} & 
    \multicolumn{2}{c}{\textbf{208.00}} \\
    \multicolumn{3}{l|}{\textbf{Median Time(s)}} & 
    \multicolumn{2}{c|}{\textbf{3.04}} & 
    \multicolumn{2}{c|}{\textbf{0.05}} & 
    \multicolumn{2}{c|}{\textbf{>1200}} & 
    \multicolumn{2}{c|}{\textbf{>1200}} & 
    \multicolumn{2}{c}{\textbf{14.96}} \\
    \bottomrule
  \end{tabular}
  }
  % \\[0.5em]
  % {\footnotesize Times in seconds (average over runs). Timeout threshold is 20 minutes.}
\end{table}

We  organize the experimental findings by four research questions (RQ1-4).

\myheading{RQ1: StageSAT vs. Incomplete Solvers (MathSAT-Large)}\emph{How does StageSAT perform relative to prior optimization-based and heuristic (incomplete) solvers on large FP benchmarks?}
We compare StageSAT with XSat, goSAT, Grater, and JFS on the 49 MathSAT-Large formulas. The results show that StageSAT attains the highest coverage and accuracy among these incomplete solvers. Table~\ref{tab:rq1-detailed} reports detailed outcomes per benchmark. StageSAT was able to return a result (SAT or unsat-guess) for 45 out of 49 instances, whereas the next-best incomplete solver (XSat) solved significantly fewer within the same time limit. In terms of SAT coverage on the \textsc{MathSAT-Large} suite, \textsc{StageSAT} found satisfying assignments for 24 of the 26 satisfiable benchmarks, outperforming all prior numeric competitors in the number of SAT instances solved. In fact, XSat comes closest in number of instances solved, but there is a crucial difference – soundness. XSat often misclassified difficult satisfiable problems as “UNSAT” when it failed to locate a solution, whereas StageSAT never falsely reported UNSAT in our tests. Specifically, XSat produced 12 incorrect UNSAT answers on benchmarks that actually have solutions (StageSAT managed to find valid models for all 12 of these). All StageSAT models were independently validated with a bit-precise checker (Z3), confirming their correctness\cite{z3py-tutorial}. This highlights that StageSAT’s optimization strategy improves not only the quantity of problems solved but also the precision of the results on challenging formulas.

Other baselines fared less well on MathSAT-Large. Grater, the state-of-the-art optimization-based solver, timed out on most large benchmarks that StageSAT solved, yielding a lower overall solve count. Moreover, Grater’s use of a fixed tolerance can lead to spurious SAT reports – in one case it reported “SAT” even though the formula was actually unsatisfiable (both Z3 and CVC5 proved it UNSAT). StageSAT avoids this pitfall by requiring a strict zero objective for SAT, and accordingly never reports SAT for an unsatisfiable instance. The heuristic fuzzer JFS and the global optimizer goSAT solved the fewest problems in this suite, struggling with the complex constraints (many of their runs resulted in timeouts or unknown results). Overall, StageSAT offers the best balance of coverage and reliability: it solves significantly more large benchmarks than JFS or goSAT, more than Grater, and matches XSat’s coverage \emph{without any} of XSat’s misclassification errors\cite{DBLP:conf/cav/FuS16}. In practical terms, StageSAT was able to solve the vast majority of these large-instance challenges (often on every trial run), whereas the other incomplete solvers were either inconsistent or outright failed on many instances. These results confirm that StageSAT’s novel staged optimization approach yields a substantial improvement in both solver effectiveness (higher solve rate) and solver soundness (no incorrect answers) for difficult floating-point problems.

\myheading{RQ2: StageSAT vs. Complete Solvers (MathSAT-Large)} \emph{How does StageSAT compare to modern complete SMT solvers on the same large benchmarks?} In this analysis we evaluate StageSAT against Bitwuzla, CVC5, Z3, and MathSAT5 on the 49 MathSAT-Large formulas, as shown in Table \ref{tab:rq2-detailed}. A key question is whether StageSAT’s heuristic approach can achieve similar coverage to these complete (bit-precise) solvers, which can in principle solve or refute any instance given unlimited time. We found that StageSAT’s coverage on large benchmarks is competitive with the best of the complete solvers. Within the 20-minute limit, StageSAT produced results for 45 out of 49 instances, slightly more than even the top-performing complete solver (Bitwuzla solved 39 in time). CVC5 solved 36, while Z3 and MathSAT5 handled fewer cases within the timeout. In other words, StageSAT’s incomplete strategy allowed it to tackle about 92\% of these hard instances, essentially matching state-of-the-art coverage on this large-scale set. This is notable because complete solvers sometimes exhaust the time on very complex formulas, whereas StageSAT either finds a model or converges on an unsatisfiability guess more quickly in most cases.

\begin{table}[htbp]
  \centering
  \caption{Comparison on MathSAT-Large benchmarks: \textsc{StageSAT} vs.\ complete solvers. Timeout as 20 mins.}
  \label{tab:rq2-detailed}
  \resizebox{\textwidth}{!}{%
  \small
  \setlength{\tabcolsep}{3pt}
  \begin{tabular}{@{}l r r | cc | cc | cc | cc | cc@{}}
    \toprule
    \multirow{2}{*}{Benchmark} & \multirow{2}{*}{Size(byte)} & \multirow{2}{*}{\#Vars} &
    \multicolumn{2}{c|}{\textsc{cvc5}} &
    \multicolumn{2}{c|}{\textsc{Bitwuzla}} &
    \multicolumn{2}{c|}{\textsc{Z3}} &
    \multicolumn{2}{c|}{\textsc{MathSAT}} &
    \multicolumn{2}{c}{\textsc{StageSAT}} \\
    \cmidrule(lr){4-5} \cmidrule(lr){6-7} \cmidrule(lr){8-9} \cmidrule(lr){10-11} \cmidrule(l){12-13}
    & & & Verdict & Time(s) & Verdict & Time(s) & Verdict & Time(s) & Verdict & Time(s) & Verdict & Time(s) \\
    \midrule
    {\footnotesize\texttt{sqrt.c.10}}                      & 21701 & 26  & sat     & 1.01   & sat     & 6.15   & sat     & 492.62 & sat     & 28.08  & sat    & 102.48 \\
    {\footnotesize\texttt{test\_v5\_r15\_vr5\_c1\_s8246}}  & 21791 & 5   & unsat   & 269.56 & unsat   & 95.73  & timeout & >1200  & unsat   & 294.65 & unsat-guess  & 10.76  \\
    {\footnotesize\texttt{test\_v5\_r15\_vr1\_c1\_s26845}} & 21811 & 5   & unsat   & 75.80  & unsat   & 47.68  & unsat   & 737.95 & unsat   & 131.37 & unsat-guess  & 7.00   \\
    {\footnotesize\texttt{test\_v5\_r15\_vr10\_c1\_s25268}}& 21818 & 5   & unsat   & 216.89 & unsat   & 156.56 & timeout & >1200  & timeout & >1200  & unsat-guess  & 8.87   \\
    {\footnotesize\texttt{test\_v5\_r15\_vr5\_c1\_s26657}} & 22070 & 5   & unsat   & 160.27 & unsat   & 63.23  & timeout & >1200  & unsat   & 177.60 & unsat-guess  & 6.70   \\
    {\footnotesize\texttt{test\_v5\_r15\_vr1\_c1\_s32559}} & 22072 & 5   & unsat   & 43.48  & unsat   & 33.09  & unsat   & 139.16 & unsat   & 24.20  & unsat-guess  & 5.78   \\
    {\footnotesize\texttt{test\_v5\_r15\_vr1\_c1\_s8236}}  & 22072 & 5   & unsat   & 49.45  & unsat   & 25.43  & unsat   & 369.09 & unsat   & 12.66  & unsat-guess  & 5.88   \\
    {\footnotesize\texttt{test\_v5\_r15\_vr5\_c1\_s23844}} & 22072 & 5   & unsat   & 266.86 & unsat   & 95.28  & timeout & >1200  & unsat   & 231.11 & unsat-guess  & 9.21   \\
    {\footnotesize\texttt{test\_v5\_r15\_vr10\_c1\_s14516}}& 22252 & 5   & unsat   & 530.11 & unsat   & 202.41 & timeout & >1200  & timeout & >1200  & unsat-guess  & 10.56  \\
    {\footnotesize\texttt{qurt.c.5}}                       & 23169 & 30  & unsat   & 6.59   & unsat   & 0.22   & unsat   & 13.44  & unsat   & 6.53   & unsat-guess  & 156.14 \\
    {\footnotesize\texttt{test\_v7\_r12\_vr5\_c1\_s29826}} & 23736 & 7   & sat     & 129.16 & sat     & 163.13 & sat     & 425.77 & sat     & 65.19  & sat    & 0.06   \\
    {\footnotesize\texttt{test\_v7\_r12\_vr10\_c1\_s15994}}& 23828 & 7   & sat     & 176.39 & sat     & 98.18  & sat     & 509.78 & sat     & 130.89 & sat    & 0.06   \\
    {\footnotesize\texttt{test\_v7\_r12\_vr10\_c1\_s30410}}& 24070 & 7   & timeout & >1200  & timeout & >1200  & timeout & >1200  & timeout & >1200  & sat    & 0.09   \\
    {\footnotesize\texttt{test\_v7\_r12\_vr5\_c1\_s14336}} & 24250 & 7   & sat     & 22.82  & sat     & 115.91 & sat     & 284.32 & sat     & 84.11  & sat    & 0.05   \\
    {\footnotesize\texttt{test\_v7\_r12\_vr5\_c1\_s8938}}  & 24251 & 7   & sat     & 19.00  & sat     & 37.72  & sat     & 105.88 & sat     & 30.66  & sat    & 0.05   \\
    {\footnotesize\texttt{test\_v7\_r12\_vr1\_c1\_s10576}} & 24274 & 7   & unsat   & 322.33 & unsat   & 170.45 & timeout & >1200  & unsat   & 206.12 & unsat-guess  & 14.56  \\
    {\footnotesize\texttt{test\_v7\_r12\_vr1\_c1\_s22787}} & 24345 & 7   & unsat   & 597.04 & unsat   & 217.96 & timeout & >1200  & unsat   & 581.35 & unsat-guess  & 11.65  \\
    {\footnotesize\texttt{test\_v7\_r12\_vr10\_c1\_s18160}}& 24437 & 7   & timeout & >1200  & timeout & >1200  & timeout & >1200  & timeout & >1200  & unsat-guess  & 21.18  \\
    {\footnotesize\texttt{test\_v7\_r12\_vr1\_c1\_s703}}   & 24441 & 7   & unsat   & 639.55 & unsat   & 234.56 & timeout & >1200  & timeout & >1200  & unsat-guess  & 16.54  \\
    {\footnotesize\texttt{sin2.c.15}}                      & 25235 & 52  & sat     & 29.95  & sat     & 159.21 & timeout & >1200  & timeout & >1200  & sat    & 3.92   \\
    {\footnotesize\texttt{gaussian.c.25}}                  & 29883 & 79  & sat     & 3.54   & sat     & 2.47   & sat     & 642.11 & sat     & 11.65  & sat    & 3.38   \\
    {\footnotesize\texttt{sqrt.c.15}}                      & 32192 & 36  & sat     & 8.95   & sat     & 1.29   & timeout & >1200  & sat     & 27.84  & sat    & 101.76 \\
    {\footnotesize\texttt{test\_v7\_r17\_vr5\_c1\_s2807}}  & 32711 & 7   & timeout & >1200  & timeout & >1200  & timeout & >1200  & timeout & >1200  & unsat-guess  & 15.21  \\
    {\footnotesize\texttt{test\_v7\_r17\_vr1\_c1\_s30331}} & 32876 & 7   & unsat   & 205.82 & unsat   & 227.17 & timeout & >1200  & unsat   & 1140.11 & unsat-guess  & 14.96  \\
    {\footnotesize\texttt{test\_v7\_r17\_vr5\_c1\_s25451}} & 32964 & 7   & timeout & >1200  & unsat   & 888.67 & timeout & >1200  & timeout & >1200  & unsat-guess  & 11.58  \\
    {\footnotesize\texttt{sin2.c.20}}                      & 33016 & 67  & sat     & 560.39 & sat     & 693.96 & timeout & >1200  & timeout & >1200  & sat    & 169.92 \\
    {\footnotesize\texttt{test\_v7\_r17\_vr10\_c1\_s8773}} & 33151 & 7   & sat     & 794.03 & sat     & 530.94 & timeout & >1200  & timeout & >1200  & sat    & 0.06   \\
    {\footnotesize\texttt{test\_v7\_r17\_vr5\_c1\_s4772}}  & 33222 & 7   & timeout & >1200  & timeout & >1200  & timeout & >1200  & timeout & >1200  & unsat-guess  & 35.20  \\
    {\footnotesize\texttt{test\_v7\_r17\_vr1\_c1\_s23882}} & 33226 & 7   & sat     & 273.29 & sat     & 340.07 & sat     & 1162.59& sat     & 524.15 & sat    & 0.05   \\
    {\footnotesize\texttt{test\_v7\_r17\_vr1\_c1\_s24331}} & 33226 & 7   & timeout & >1200  & unsat   & 787.68 & timeout & >1200  & timeout & >1200  & unsat-guess  & 11.63  \\
    {\footnotesize\texttt{test\_v7\_r17\_vr10\_c1\_s3680}} & 33335 & 7   & timeout & >1200  & unsat   & 1188.71& timeout & >1200  & timeout & >1200  & unsat-guess  & 13.26  \\
    {\footnotesize\texttt{test\_v7\_r17\_vr10\_c1\_s18654}}& 33410 & 7   & sat     & 1126.63& sat     & 588.61 & timeout & >1200  & timeout & >1200  & sat    & 0.05   \\
    {\footnotesize\texttt{sin.c.25}}                       & 40536 & 81  & sat     & 802.21 & sat     & 874.55 & timeout & >1200  & timeout & >1200  & sat    & 294.71 \\
    {\footnotesize\texttt{sin2.c.25}}                      & 40747 & 82  & sat     & 713.26 & sat     & 902.61 & timeout & >1200  & timeout & >1200  & sat    & 8.62   \\
    {\footnotesize\texttt{sqrt.c.20}}                      & 46804 & 63  & sat     & 0.96   & sat     & 5.96   & sat     & 721.16 & sat     & 18.39  & sat    & 498.60 \\
    {\footnotesize\texttt{sqrt.c.25}}                      & 46804 & 63  & sat     & 0.96   & sat     & 5.95   & sat     & 717.69 & sat     & 18.42  & sat    & 581.95 \\
    {\footnotesize\texttt{qurt.c.10}}                      & 47946 & 60  & unsat   & 6.59   & unsat   & 0.21   & unsat   & 23.86  & unsat   & 6.62   & unsat-guess  & 393.02 \\
    {\footnotesize\texttt{qurt.c.15}}                      & 73125 & 90  & unsat   & 0.27   & unsat   & 0.02   & unsat   & 1.93   & unsat   & 6.61   & unsat-guess  & 728.38 \\
    {\footnotesize\texttt{gaussian.c.75}}                  & 89686 & 229 & sat     & 13.46  & sat     & 6.57   & sat     & 429.71 & sat     & 56.89  & sat    & 36.45  \\
    {\footnotesize\texttt{qurt.c.20}}                      & 93126 & 114 & unsat   & 0.30   & unsat   & 0.01   & timeout & >1200  & unsat   & 7.25   & timeout& >1200     \\
    {\footnotesize\texttt{qurt.c.25}}                      & 93126 & 114 & unsat   & 0.32   & unsat   & 0.01   & timeout & >1200  & unsat   & 7.21   & timeout& >1200     \\
    {\footnotesize\texttt{sin2.c.75}}                      & 119790& 231 & timeout & >1200  & timeout & >1200  & error   & --     & timeout & >1200  & sat    & 363.37 \\
    {\footnotesize\texttt{sin.c.75}}                       & 119794& 231 & timeout & >1200  & timeout & >1200  & error   & --     & timeout & >1200  & sat    & 660.70 \\
    {\footnotesize\texttt{gaussian.c.125}}                 & 150792& 379 & sat     & 20.54  & sat     & 36.94  & timeout & >1200  & sat     & 338.55 & sat    & 90.62  \\
    {\footnotesize\texttt{sin.c.125}}                      & 200503& 381 & timeout & >1200  & timeout & >1200  & error   & --     & timeout & >1200  & sat    & 379.39 \\
    {\footnotesize\texttt{sin2.c.125}}                     & 200503& 381 & timeout & >1200  & timeout & >1200  & error   & --     & timeout & >1200  & sat    & 313.85 \\
    {\footnotesize\texttt{gaussian.c.175}}                 & 210711& 529 & sat     & 460.36 & sat     & 745.12 & timeout & >1200  & timeout & >1200  & sat    & 273.68 \\
    {\footnotesize\texttt{sin2.c.175}}                     & 280962& 531 & timeout & >1200  & timeout & >1200  & error   & --     & timeout & >1200  & timeout& >1200     \\
    {\footnotesize\texttt{sin.c.175}}                      & 280984& 531 & timeout & >1200  & timeout & >1200  & error   & --     & timeout & >1200  & timeout& >1200     \\
    \midrule
    \multicolumn{3}{l|}{\textbf{SAT/UNSAT Coverage}} & 
    \multicolumn{2}{c|}{\textbf{73.5\%}} & 
    \multicolumn{2}{c|}{\textbf{79.6\%}} & 
    \multicolumn{2}{c|}{\textbf{32.7\%}} & 
    \multicolumn{2}{c|}{\textbf{53.1\%}} & 
    \multicolumn{2}{c}{\textbf{91.8\%}} \\
    \multicolumn{3}{l|}{\textbf{Timeout Rate}} & 
    \multicolumn{2}{c|}{\textbf{26.5\%}} & 
    \multicolumn{2}{c|}{\textbf{20.4\%}} & 
    \multicolumn{2}{c|}{\textbf{55.1\%}} & 
    \multicolumn{2}{c|}{\textbf{46.9\%}} & 
    \multicolumn{2}{c}{\textbf{8.2\%}} \\
    \multicolumn{3}{l|}{\textbf{Average Time(s)}} & 
    \multicolumn{2}{c|}{\textbf{492.82}} & 
    \multicolumn{2}{c|}{\textbf{443.89}} & 
    \multicolumn{2}{c|}{\textbf{911.09}} & 
    \multicolumn{2}{c|}{\textbf{648.33}} & 
    \multicolumn{2}{c}{\textbf{208.00}} \\
    \multicolumn{3}{l|}{\textbf{Median Time(s)}} & 
    \multicolumn{2}{c|}{\textbf{269.56}} & 
    \multicolumn{2}{c|}{\textbf{170.45}} & 
    \multicolumn{2}{c|}{\textbf{>1200}} & 
    \multicolumn{2}{c|}{\textbf{581.35}} & 
    \multicolumn{2}{c}{\textbf{14.96}} \\
    \bottomrule
  \end{tabular}
  }
  % \\[0.5em]
  % {\footnotesize Times in seconds (average over runs). Timeout threshold is 20 minutes.}
\end{table}
\clearpage

\begin{table}[H]
  \centering
  \small
  \setlength{\tabcolsep}{3pt}
  \caption{Comparison on \textsc{MathSAT-Large} benchmarks where \textsc{Bitwuzla} timeout: \textsc{StageSAT} vs.\ \textsc{Bitwuzla}. Timeout as 48 hours.}
  \label{tab:rq2-detailed-stagesat-bitwuzla}
  \begin{tabular}{@{}l r r | cc | cc@{}}
    \toprule
    \multirow{2}{*}{Benchmark} & \multirow{2}{*}{Size(byte)} & \multirow{2}{*}{\#Vars} &
    \multicolumn{2}{c|}{\textsc{StageSAT}} &
    \multicolumn{2}{c}{\textsc{Bitwuzla}} \\
    \cmidrule(lr){4-5} \cmidrule(l){6-7}
    & & & Verdict & Time(s) & Verdict & Time(s) \\
    \midrule
    {\footnotesize\texttt{test\_v7\_r12\_vr10\_c1\_s30410}}& 24070 & 7   & sat    & 0.09   & sat & 3042.79     \\
    {\footnotesize\texttt{test\_v7\_r12\_vr10\_c1\_s18160}}& 24437 & 7   & unsat-guess  & 21.18  & unsat & 7347.42     \\
    {\footnotesize\texttt{test\_v7\_r17\_vr5\_c1\_s2807}}  & 32711 & 7   & unsat-guess  & 15.21  & unsat & 1466.26     \\
    {\footnotesize\texttt{test\_v7\_r17\_vr5\_c1\_s4772}}  & 33222 & 7   & unsat-guess  & 35.20  & unsat & 93950.26     \\
    {\footnotesize\texttt{sin2.c.75}}                      & 119790& 231 & sat    & 363.37 & sat & 4180.57     \\
    {\footnotesize\texttt{sin.c.75}}                       & 119794& 231 & sat    & 660.70 & sat & 9361.38     \\
    {\footnotesize\texttt{sin.c.125}}                      & 200503& 381 & sat    & 379.39 & timeout & >48hours     \\
    {\footnotesize\texttt{sin2.c.125}}                     & 200503& 381 & sat    & 313.85 & sat & 105514.52     \\
    {\footnotesize\texttt{sin2.c.175}}                     & 280962& 531 & timeout & >48hours     & sat & 151645.17     \\
    {\footnotesize\texttt{sin.c.175}}                      & 280984& 531 & timeout & >48hours     & sat & 68898.42     \\
    \bottomrule
  \end{tabular}
\end{table}

In terms of performance, \textsc{StageSAT} also shows advantages on satisfiable instances: it typically finds solutions significantly faster than complete solvers. Many satisfiable \textsc{MathSAT-Large} benchmarks that Bitwuzla or CVC5 only solve near the timeout are solved by \textsc{StageSAT} in a fraction of the time (often $5$–$10\times$ faster, and on the largest satisfiable cases over $10\times$ faster in our tests). We also ran stress tests with \textsc{StageSAT} and Bitwuzla under an extended $48$-hour timeout to evaluate their behavior on the largest benchmarks where Bitwuzla times out, with results summarized in Table~\ref{tab:rq2-detailed-stagesat-bitwuzla}.
For example, on one particularly difficult satisfiable formula (sin.c.125), StageSAT succeeded quickly while Bitwuzla \emph{failed to find any model even with 48 hours of searching}. This underscores StageSAT’s strength in navigating the search space for models efficiently. On the other hand, complete solvers have an edge in proving unsatisfiability. By design, StageSAT cannot provide formal UNSAT proofs – it will output unsat-guess when its optimization process concludes that no better (lower) objective can be found, but this is a heuristic indication. We examined all cases: every benchmark that was proven UNSAT by any complete solver was classified as unsat-guess by StageSAT. Crucially, we observed no discrepancies – StageSAT did not label any instance as satisfiable that a complete solver proved unsatisfiable, and whenever StageSAT gave an unsat-guess, the instance indeed had no solution according to the complete solvers. This empirical agreement suggests that StageSAT’s UNSAT guesses were reliable on MathSAT-Large: they acted as a correct heuristic proxy for true unsatisfiability in all tested cases. However, since these are not formal proofs, we do not count them as proven UNSAT; instead, they demonstrate that StageSAT can recognize unsolvable instances with reasonable confidence. On the four hardest cases where StageSAT timed out (no result), the complete solvers eventually determined two to be satisfiable and two unsatisfiable after extended runs. This highlights a limitation: for the very largest or trickiest satisfiable formulas, StageSAT may also struggle (as it did on two extremely large satisfiable benchmarks that required more than 20 minutes). Conversely, for some tough UNSAT cases, complete solvers’ heavy-duty reasoning succeeded where StageSAT’s heuristic approach did not finish in time.

Despite these few limitations, StageSAT performs well against the complete solvers. It matched or exceeded their solve counts within the standard time limit and delivered solutions faster on the instances it could solve. This indicates StageSAT can serve as a practical complement to traditional SMT solvers. In a usage scenario, one might run StageSAT alongside a complete solver: StageSAT will quickly find a model if one exists for hard satisfiable problems (saving potentially hours of search), while for the few instances that require a proof of unsatisfiability, a complete solver can take over. StageSAT’s ability to handle large benchmarks competitively demonstrates the benefit of its approach – bridging numeric optimization with SMT – in achieving both scalability and empirical correctness on floating-point constraints.

\myheading{RQ3: Ablation Study of StageSAT’s Design}
\emph{
Are all three stages and key design components of StageSAT necessary for its performance?} StageSAT’s solving process comprises three sequential stages (S1, S2, S3) with different objective formulations, plus additional techniques like the projection term in S1 and a clause-wise ULP aggregation in later stages. To understand the contribution of each part, we performed an ablation study: running modified versions of StageSAT with one component removed or altered, and measuring the impact on MathSAT-Large results. Table \ref{tab:rq3-ablation} presents the outcome counts for each variant. The full StageSAT (S1–S2–S3 as designed) solved 45 of 49 large benchmarks (24 reported SAT, 21 unsat-guess) with an average solve time of 208 s per instance. We use this as the baseline for comparison.

\begin{table}[htbp]
  \centering
  \small
  \setlength{\tabcolsep}{6pt}
  \caption{Ablation summary on MathSAT-Large benchmarks. Timeout as 20 mins.}
  \label{tab:rq3-ablation}
  \begin{tabular}{l r r r r}
    \hline
    Variant & SAT (\#) & UNSAT (\#) & Timeout (\#) & Avg.\ time (s) \\
    \hline
    Full S1--S2--S3                        & 24 & 21 & 4 & 208.00 \\
    No S1 (start S2 directly)              & 13 & 22 & 14 & 409.07 \\
    No S3 (S1--S2 only)                    & 22 & 23 & 4 & 262.61 \\
    S1 without projection term             & 21 & 24 & 4 & 291.82 \\
    S1 with absolute residuals             & 18 & 27 & 4 & 196.05 \\
    S2/S3 without clause-wise ULP product  & 22 & 21 & 6 & 264.06 \\
    \hline
  \end{tabular}
  % \\[0.5em]
  % {\footnotesize Variants: remove or modify one stage at a time.}
\end{table}

Removing Stage 1 and starting directly with Stage 2 had a drastic effect on performance. This "No S1" variant solved only 35/49 instances (a drop from 45) and suffered 14 timeouts (versus 4 in full StageSAT). In particular, the number of satisfiable instances found fell sharply (from 24 down to just 13). The average runtime on solved cases also nearly doubled (409 s vs 208 s). This shows that Stage 1 is crucial for scalability – its projection-aided, squared-residual objective quickly guides the solver toward feasible regions. Without S1’s fast coarse guidance, the solver struggled and timed out much more frequently. Intuitively, S1’s lightweight objective (based on magnitudes of residuals) is much cheaper to evaluate than the more precise ULP-based objectives in S2 and S3, so it can explore the search space broadly and find a promising region quickly. Conclusion: Stage 1 is indispensable for reaching high coverage on large problems.

Conceptually, Stage 2 is critical in regimes where purely squared-magnitude objectives get stuck, such as subnormal regions and underflow.
The \textsc{MathSAT-Large} set doesn't contain such stress cases, which is why we do not include a ``No S2'' configuration in this ablation study.
In contrast, the JFS and Grater benchmarks include many formulas with similar structure, and on those suites Stage 2 is essential for steering \textsc{StageSAT} toward correct models.

Next, removing Stage 3 had a more subtle but important impact. The “No S3” variant solved the same total number of instances (45/49), \emph{but} it found fewer SAT solutions (22 vs 24) and correspondingly classified more cases as unsat-guess (23 vs 21). In fact, two benchmarks that full StageSAT successfully solved as SAT were missed by the S1–S2-only solver, which converged to a non-zero minimum and reported them (incorrectly) as unsatisfiable. This indicates that Stage 3 is essential for correctness on edge cases – it performs a discrete, high-precision search over the FP lattice that can “close the gap” when the continuous optimization in S2 stalls. Without S3, StageSAT’s SAT coverage would drop and it would misclassify some satisfiable problems as unsat-guess. Conclusion: Stage 3 is necessary to achieve StageSAT’s 100\% soundness for SAT results, by catching those tricky cases where numerical smoothing alone isn’t enough.

We also tested variations in the \emph{within-stage techniques}. Removing the orthogonal projection term from Stage 1 (while still doing S1–S2–S3) kept the total solved count at 45, but StageSAT’s SAT find rate worsened slightly (21 SAT vs 24) and the average solve time increased to \textasciitilde{}292 s. This suggests the projection term – which imposes a partial monotone descent property – indeed helps the solver find more models (3 extra in the full version) and do so faster. Its absence likely caused the optimizer to sometimes get bogged down on flat or misleading surfaces, turning a few would-be SAT cases into unsat-guesses and generally slowing progress. Replacing S1’s squared-residual objective with a simpler absolute residual objective had a mixed effect: it slightly improved runtime (196 s average) but at a steep cost in SAT solutions (only 18 SAT found, with 27 unsat-guess). This trade-off indicates that while absolute residuals might speed up convergence in some cases, they provide a weaker guidance toward exact solutions, causing StageSAT to miss many satisfiable cases that the squared residual version would solve. Lastly, disabling the clause-wise ULP product in Stages 2–3 (an aggregation scheme that emphasizes satisfying all constraints together) reduced StageSAT’s overall solves to 43/49 and led to 6 timeouts (versus 4 in full StageSAT). It also slightly diminished SAT coverage (22 SAT vs 24 in full StageSAT). This indicates that the clause-wise product, though a simple heuristic to combine per-clause errors, helps StageSAT prioritize assignments that satisfy \emph{every} constraint, thereby avoiding some timeouts and finding a couple more models.

In summary, the ablation study confirms that each stage and each design choice contributes meaningfully to StageSAT’s success. Stage 1’s projection-aided coarse search is vital for scale and speed; Stage 2’s ULP-based continuous optimization is critical for navigating tricky FP regions; Stage 3’s discrete refinement is crucial for final correctness; and the added touches (projection term, squared metrics, ULP aggregation) further improve both performance and result quality. Removing any one of these either hurts StageSAT’s ability to solve tough instances or causes it to miss solutions, validating the full S1–S2–S3 design as necessary to achieve the reported high coverage and accuracy.

\begin{table}[tbp]
  \centering
  \scriptsize
  \setlength{\tabcolsep}{3pt}
  \renewcommand{\arraystretch}{0.9}
  \caption{Summary on MathSAT-Large, MathSAT-Middle, MathSAT-Small, JFS, and Grater benchmarks.\\
\#Unk/UG = \#Unknown/Unsat-Guess. Timeout as 20 mins.}
  \label{tab:summary-large-middle}
  \begin{tabular}{@{}l l r r r r r r r r@{}}
    \toprule
    Dataset & Solver & \#SAT & \#UNSAT & \#Timeout & \#Unk/UG & \#Error & SAT Coverage & Avg.\ time(s) & Med.\ time(s) \\
    \midrule
    \multicolumn{10}{l}{\textbf{MathSAT-Large   (49 benchmarks, 26 SAT, 23 UNSAT)}} \\
    \midrule[0.5pt]
     & \textsc{cvc5}     & 19 & 17 & 13 & 0  & 0 & 73.1\% & 492.82 & 269.56 \\
     & \textsc{Bitwuzla} & 19 & 20 & 10 & 0  & 0 & 73.1\% & 443.89 & 170.45 \\
     & \textsc{Z3}       & 10 &  6 & 27 & 0  & 6 & 38.5\% & 911.09 & >1200 \\
     & \textsc{MathSAT}  & 12 & 14 & 23 & 0  & 0 & 46.2\% & 648.33 & 581.35 \\
     & \textsc{XSat}     & 12 & 23 &  2 & 0  & 0 & 46.2\% & 134.02 & 3.04 \\
     & \textsc{goSAT}    &  8 &  0 &  0 & 41 & 0 & 30.8\% & 8.27 & 0.05 \\
     & \textsc{Grater}   & 16 &  0 & 25 & 0  & 4 & 57.7\% & 819.24 & >1200 \\
     & \textsc{JFS}      &  4 &  0 & 45 & 0  & 0 & 15.4\% & 1102.46 & >1200 \\
     & \textsc{StageSAT} & 24 & 0 &  4 & 21  & 0 & 92.3\% & 208.00 & 14.96 \\
    \midrule[0.5pt]
    \multicolumn{10}{l}{\textbf{MathSAT-Middle  (35 benchmarks, 29 SAT, 6 UNSAT)}} \\
    \midrule[0.5pt]
     & \textsc{cvc5}     & 29 & 6 & 0 & 0 & 0 & 100.0\% & 94.65 & 39.24 \\
     & \textsc{Bitwuzla} & 29 & 6 & 0 & 0 & 0 & 100.0\% & 31.10 & 22.59 \\
     & \textsc{Z3}       & 19 & 5 & 6 & 0 & 5 & 65.5\% & 370.31 & 323.14 \\
     & \textsc{MathSAT}  & 28 & 6 & 1 & 0 & 0 & 96.6\% & 77.19 & 33.47 \\
     & \textsc{XSat}     & 29 & 6 & 0 & 0 & 0 & 100.0\% & 1.06 & 0.14 \\
     & \textsc{goSAT}    & 23 & 0 & 0 & 12& 0 & 79.3\% & 0.26 & 0.02 \\
     & \textsc{Grater}   & 26 & 0 & 9 & 0 & 0 & 89.7\% & 337.13 & 0.18 \\
     & \textsc{JFS}      &  9 & 0 & 26& 0 & 0 & 31.0\% & 892.06 & >1200 \\
     & \textsc{StageSAT} & 29 & 0 & 0 & 6 & 0 & 100.0\% & 8.59 & 0.18 \\
    \midrule[0.5pt]
    \multicolumn{10}{l}{\textbf{MathSAT-Small   (130 benchmarks, 63 SAT, 67 UNSAT)}} \\
    \midrule[0.5pt]
     & \textsc{cvc5}     & 63 & 55 & 12 & 0 & 0 & 100.0\% & 136.44 & 3.34 \\
     & \textsc{Bitwuzla} & 63 & 64 & 3 & 0 & 0 & 100.0\% & 35.64 & 1.93 \\
     & \textsc{Z3}       & 56 & 53 & 13 & 0 & 8 & 88.9\% & 194.18 & 24.79 \\
     & \textsc{MathSAT}  & 63 & 51 & 16 & 0 & 0 & 100.0\% & 168.37 & 7.31 \\
     & \textsc{XSat}     & 62 & 67 & 0 & 0 & 1 & 98.4\% & 0.46 & 0.14 \\
     & \textsc{goSAT}    & 58 & 0 & 0 & 72 & 0 & 92.1\% & 0.07 & 0.02 \\
     & \textsc{Grater}   & 63 & 0 & 38 & 0 & 28 & 100.0\% & 471.50 & 0.81 \\
     & \textsc{JFS}      & 53 & 0 & 77 & 0 & 0 & 84.1\% & 717.96 & >1200 \\
     & \textsc{StageSAT} & 63 & 0 & 0 & 67 & 0 & 100.0\% & 1.64 & 0.66 \\
    \midrule[0.5pt]
    \multicolumn{10}{l}{\textbf{JFS             (111 benchmarks, 111 SAT, 0 UNSAT)}} \\
    \midrule[0.5pt]
     & \textsc{cvc5}     & 104 & 0 & 6 & 0 & 1 & 93.7\% & 110.87 & 4.74 \\
     & \textsc{Bitwuzla} & 108 & 0 & 3 & 0 & 0 & 97.3\% & 103.65 & 3.24 \\
     & \textsc{Z3}       & 85 & 0 & 18 & 0 & 8 & 76.6\% & 327.55 & 4.27 \\
     & \textsc{MathSAT}  & 102 & 0 & 9 & 0 & 0 & 91.9\% & 149.23 & 11.60 \\
     & \textsc{XSat}     & 101 & 4 & 0 & 0 & 6 & 91.0\% & 13.04 & 0.14 \\
     & \textsc{goSAT}    & 88 & 0 & 0 & 19 & 4 & 79.3\% & 2.21 & 0.03 \\
     & \textsc{Grater}   & 107 & 0 & 4 & 0 & 0 & 96.4\% & 71.86 & 0.06 \\
     & \textsc{JFS}      & 69 & 0 & 42 & 0 & 0 & 62.2\% & 458.33 & 0.72 \\
     & \textsc{StageSAT} & 111 & 0 & 0 & 0 & 0 & 100.0\% & 13.10 & 0.05 \\
    \midrule[0.5pt]
    \multicolumn{10}{l}{\textbf{Grater          (118 benchmarks, 118 SAT, 0 UNSAT)}} \\
    \midrule[0.5pt]
     & \textsc{cvc5}     & 118 & 0 & 0 & 0 & 0 & 100.0\% & 37.62 & 2.13 \\
     & \textsc{Bitwuzla} & 117 & 0 & 1 & 0 & 0 & 99.2\% & 64.30 & 10.61 \\
     & \textsc{Z3}       & 103 & 0 & 12 & 0 & 3 & 87.3\% & 241.12 & 64.73 \\
     & \textsc{MathSAT}  & 111 & 0 & 7 & 0 & 0 & 94.1\% & 131.05 & 25.29 \\
     & \textsc{XSat}     & 115 & 2 & 0 & 0 & 1 & 97.5\% & 1.62 & 0.14 \\
     & \textsc{goSAT}    & 92 & 0 & 0 & 26 & 0 & 78.0\% & 0.70 & 0.02 \\
     & \textsc{Grater}   & 118 & 0 & 0 & 0 & 0 & 100.0\% & 4.88 & 0.41 \\
     & \textsc{JFS}      & 40 & 0 & 78 & 0 & 0 & 33.9\% & 793.65 & >1200 \\
     & \textsc{StageSAT} & 118 & 0 & 0 & 0 & 0 & 100.0\% & 10.34 & 0.18 \\
    \bottomrule
  \end{tabular}
  % \\[0.5em]
  % {\footnotesize 
  %   \#Unk/UG denotes \#Unknown/Unsat-Guess. \\
  %   \#SAT and \#UNSAT count correctly classified instances only. \\
  %   Times in seconds (average over runs). Timeout threshold is 20 minutes.
  %   }
\end{table}

\myheading{RQ4: Overall Performance Across All Benchmark Suites}\emph{What is StageSAT’s overall performance when considering all benchmark categories (small, medium, large, and different sources) and how does it compare to other solvers across the board?} To answer this, we aggregated results from all five benchmark sets in Table \ref{tab:summary-large-middle}: MathSAT-Large, MathSAT-Middle, MathSAT-Small, JFS, and Grater. StageSAT’s trends observed in the large suite extend to the others, showing a consistently strong performance. First, StageSAT was able to complete every single instance in the MathSAT-Small, MathSAT-Middle, JFS, and Grater sets within the timeout. In those four suites (totaling over 280 formulas), StageSAT had zero timeouts and solved all satisfiable instances, achieving 100\% coverage. The MathSAT-Large suite remained the only source of timeouts for StageSAT (4 out of 49, as discussed). This means that across all benchmarks tested, StageSAT solved 45 (Large) + (all 35 Middle) + (all 130 Small) + (all 111 JFS) + (all 118 Grater) = 439 problems, missing only the 4 hardest large cases. By contrast, other incomplete solvers often left many problems unsolved (especially on the larger instances), and even complete solvers occasionally timed out on some medium or small benchmarks with complex formulas. StageSAT’s ability to handle \emph{every} instance in the easier suites underscores its robustness and efficiency even on simpler or moderate problems.

Looking at SAT coverage in particular – i.e. the fraction of truly satisfiable benchmarks for which a solver can find a model – StageSAT was the top performer on every suite against the incomplete solvers. It consistently identified the most solutions. For example, in the Grater benchmarks (many of which are satisfiable tricky cases), StageSAT found solutions for all satisfiable instances, whereas XSat, goSAT, or JFS missed some and/or gave up early. In the JFS suite, which was tailored to a fuzzing approach, StageSAT also managed to solve all satisfiable cases, effectively equaling or surpassing JFS’s coverage but with a systematic method. Meanwhile, StageSAT’s unsat-guess heuristic did not lead it astray on the smaller benchmarks either: on all problems where complete solvers reported UNSAT, StageSAT either also reported unsat-guess or, in many cases, was able to quickly decide unsat-guess before the complete solver timed out. There were no instances where StageSAT’s answer disagreed with the known ground truth. This level of agreement gives confidence that StageSAT’s approach scales down well in addition to scaling up.

Comparing StageSAT to complete solvers across all sets, we observe a complementary profile. Complete SMT solvers like Bitwuzla and CVC5 maintain perfect soundness (never guessing UNSAT) and can eventually solve every problem given enough time, but they often require significantly more time on satisfiable cases. StageSAT, on the other hand, excels in speed: across the board it dramatically reduced solve times on satisfiable benchmarks, with 3×–100× lower mean runtimes and one to two orders of magnitude lower median runtimes compared to the complete solvers in each category. Especially on the harder end of MathSAT-Large and on the JFS/Grater sets (which contain complex mathematical constraints), StageSAT’s numeric search finds solutions much faster than exhaustive search, while the complete solvers tend to hit many timeouts or long runs. In terms of coverage, StageSAT is competitive: it solved as many or more instances as the best complete solver on Large, and it solved \emph{every} instance on the other suites, whereas some complete solvers struggled with a handful of those (due to the exponential blow-up of bit-blasting on certain formulas). For example, on the Grater benchmarks, StageSAT had no timeouts, whereas some complete solvers timed out on a few; similarly for the JFS suite. This suggests that for these types of floating-point problems, StageSAT can serve as an effective front-line solver to quickly catch the easy satisfiable cases and many of the hard ones, thereby alleviating the load on complete solvers.

\section{Related Work}

\paragraph{Mainstream SMT Solvers for Floating-Point.}
Modern SMT solvers such as Z3~\cite{DBLP:conf/tacas/MouraB08}, CVC5~\cite{DBLP:conf/tacas/BarbosaBBKLMMMN22}, MathSAT5~\cite{mathsat5, mathsat-fmcad}, and Bitwuzla~\cite{DBLP:conf/cav/NiemetzP23} represent the state-of-the-art in bit-precise reasoning for floating-point constraints. These solvers reduce floating-point formulas to lower-level theories (typically bit-vectors via bit-blasting) and leverage DPLL(T)/CDCL(T)\cite{DBLP:conf/cav/dpll1, DBLP:journals/jacm/dpll2, cdcl1, dpll3} search to ensure completeness. This approach has proven highly effective on many benchmarks – indeed, Bitwuzla and CVC5 have dominated recent SMT-COMP competitions in floating-point divisions\cite{SMTCOMP2025}. However, bit-precise methods can struggle with the enormous search space induced by low-level encodings, especially for complex numerical constraints. In contrast, StageSAT takes an alternative route: rather than exhaustive logical reasoning at the bit-level, it harnesses numerical optimization techniques\cite{weak-distance, numerical, numerical2} to navigate the search space. This strategy is fundamentally different and complementary to bit-blasting approaches. StageSAT avoids enumerating bit patterns explicitly, potentially scaling to constraints that are intractable for traditional solvers, while still ultimately producing bit-precise satisfying assignments. Unlike Bitwuzla’s recent incorporation of bit-level local search heuristics\cite{DBLP:conf/aaai/FrohlichBWH15}, StageSAT operates at the numeric level, using a staged optimization process to guide the solver toward a solution. This novel perspective allows StageSAT to tackle floating-point constraints from a fresh angle, without directly competing with the intricate but heavyweight bit-precise procedures of Z3, CVC5, MathSAT, and Bitwuzla.

\paragraph{Heuristic and Fuzzing-Based Solvers.}
An orthogonal line of work has explored heuristic methods for SMT solving. Notably, JFS (the "Just Fuzz It" Solver)~\cite{jfs} applies coverage-guided fuzzing~\cite{afl, aflfast} to floating-point satisfiability, randomly mutating inputs using feedback from solver executions. This approach can quickly find solutions for certain FP problems that confound traditional solvers. However, JFS is neither complete nor optimization-based: it cannot prove unsatisfiability and lacks a principled objective beyond code coverage heuristics. StageSAT is markedly different: we define a clear mathematical objective whose minimization corresponds to satisfying the formula. While JFS demonstrates the value of unconventional search strategies, StageSAT provides a more systematic approach grounded in optimization. Unlike JFS's random exploration, StageSAT leverages its optimization framework to navigate toward models with principled guidance rather than guessing.

\paragraph{Optimization-Based Floating-Point Solvers.}
Our work is most directly inspired by prior attempts to solve floating-point constraints via mathematical optimization. XSat~\cite{DBLP:conf/cav/FuS16} pioneered this direction by translating constraints into a real-valued optimization problem checked against IEEE-754 semantics~\cite{ieee754}. This yielded a fast solver but could miss corner cases where real and floating-point solutions diverge. GoSAT~\cite{DBLP:conf/fmcad/KhadraSK17} formulated floating-point satisfiability as a global optimization problem, efficiently finding models for difficult FP problems but remaining incomplete and vulnerable to discontinuous regions. More recently, Grater~\cite{grater} refined this paradigm using carefully constructed continuous objectives and gradient-based methods, achieving impressive performance matching or surpassing Bitwuzla and CVC5 on several benchmarks. parSAT~\cite{parsat} extends this work by parallelizing multiple stochastic optimization methods~\cite{basinhopping, crs2, isres} on multi-core CPUs. Despite these advances, all prior optimization-based solvers operate in a single phase: they reduce the entire formula to one monolithic objective and apply an off-the-shelf optimizer. This one-stage approach leaves them vulnerable to local minima and irregular constraint landscapes with piecewise-defined semantics or abrupt discontinuities.

StageSAT distinguishes itself by its \emph{staged optimization} strategy, which to our knowledge has not been explored in prior SMT solvers. Instead of a single optimization run, StageSAT breaks the solving process into multiple stages, each solving a sub-problem that incrementally moves closer to satisfiability. This enables StageSAT to surmount hurdles that tripped up earlier tools: it first solves a relaxed version of the constraints (smoothing out non-linear or discontinuous behaviors), then progressively re-introduces full floating-point semantics. This avoids local minima that one-shot methods fall into and handles discontinuities by isolating them in specific stages. None of XSat, goSAT, or Grater employs such a multi-phase scheme. This multi-stage design balances the smoothness of numeric optimization with bit-level precision, resulting in a solver that is more robust on certain hard benchmarks than prior techniques.

\section{Conclusion}

We introduced StageSAT, a solver for satisfiability in the SMT logic QF\_FP, which addresses the challenges of floating-point constraints with a novel three-stage architecture. StageSAT’s pipeline combines (i) fast projection-aided descent, (ii) an ULP$^2$-shaped objective optimization, and (iii) an n-ULP lattice refinement. By integrating numeric search with a final bit-level refinement, StageSAT efficiently navigates the floating-point search space while ensuring solution correctness. This design is both accurate and practically effective, overcoming key limitations of prior numeric solvers that often suffer from imprecise search or incomplete results. In particular, StageSAT’s lattice-refinement phase guarantees high precision, making it novel compared to earlier numeric SMT solvers that could not always ensure correct or complete results. Overall, StageSAT’s multi-phase approach represents a principled advancement in floating-point SMT solving, combining the strengths of continuous optimization and discrete search in a way not seen in previous tools.

Our experimental results validate StageSAT's approach: it matched or exceeded the SAT recall of all other solvers in our benchmarks. It also solved more large, challenging benchmarks than any competing solver, attaining the highest coverage on the most complex problem sets. At the same time, StageSAT exhibited the lowest timeout rate among all tools evaluated, indicating its superior reliability on hard constraints. For example, on the largest benchmark category, StageSAT solved the most instances with the fewest timeouts, whereas the best alternative left several satisfiable problems unsolved. These concrete results demonstrate that StageSAT delivers robust and accurate performance in practice, substantially outperforming both prior numeric approaches and state-of-the-art bit-precise solvers on difficult floating-point problems. In conclusion, the StageSAT solver is novel, accurate, and effective – it advances the state of the art in QF\_FP satisfiability by solving complex floating-point constraints with high precision and practical efficiency in our evaluation.

\clearpage

\bibliographystyle{ACM-Reference-Format}

\bibliography{references}

@inproceedings{DBLP:conf/cav/FuS16,
  author       = {Zhoulai Fu and
                  Zhendong Su},
  editor       = {Swarat Chaudhuri and
                  Azadeh Farzan},
  title        = {XSat: {A} Fast Floating-Point Satisfiability Solver},
  booktitle    = {Computer Aided Verification - 28th International Conference, {CAV}
                  2016, Toronto, ON, Canada, July 17-23, 2016, Proceedings, Part {II}},
  series       = {Lecture Notes in Computer Science},
  volume       = {9780},
  pages        = {187--209},
  publisher    = {Springer},
  year         = {2016},
  url          = {https://doi.org/10.1007/978-3-319-41540-6\_11},
  doi          = {10.1007/978-3-319-41540-6\_11},
  bibsource    = {dblp computer science bibliography, https://dblp.org}
}

@inproceedings{DBLP:conf/fmcad/KhadraSK17,
  author       = {M. Ammar {Ben Khadra} and
                  Dominik Stoffel and
                  Wolfgang Kunz},
  editor       = {Daryl Stewart and
                  Georg Weissenbacher},
  title        = {goSAT: Floating-point satisfiability as global optimization},
  booktitle    = {2017 Formal Methods in Computer Aided Design, {FMCAD} 2017, Vienna,
                  Austria, October 2-6, 2017},
  pages        = {11--14},
  publisher    = {{IEEE}},
  year         = {2017},
  url          = {https://doi.org/10.23919/FMCAD.2017.8102235},
  doi          = {10.23919/FMCAD.2017.8102235},
  bibsource    = {dblp computer science bibliography, https://dblp.org}
}

@article{grater,
  author       = {Qian Chen and
                  Chenqi Cui and
                  Fengjuan Gao and
                  Yu Wang and
                  Ke Wang and
                  Linzhang Wang},
  title        = {Solving Floating-Point Constraints with Continuous Optimization},
  journal      = {Proc. {ACM} Program. Lang.},
  volume       = {9},
  number       = {{PLDI}},
  pages        = {725--747},
  year         = {2025},
  url          = {https://doi.org/10.1145/3729279},
  doi          = {10.1145/3729279},
  bibsource    = {dblp computer science bibliography, https://dblp.org}
}

@inproceedings{jfs,
  author       = {Daniel Liew and
                  Cristian Cadar and
                  Alastair F. Donaldson and
                  J. Ryan Stinnett},
  editor       = {Marlon Dumas and
                  Dietmar Pfahl and
                  Sven Apel and
                  Alessandra Russo},
  title        = {Just fuzz it: solving floating-point constraints using coverage-guided
                  fuzzing},
  booktitle    = {Proceedings of the {ACM} Joint Meeting on European Software Engineering
                  Conference and Symposium on the Foundations of Software Engineering,
                  {ESEC/SIGSOFT} {FSE} 2019, Tallinn, Estonia, August 26-30, 2019},
  pages        = {521--532},
  publisher    = {{ACM}},
  year         = {2019},
  url          = {https://doi.org/10.1145/3338906.3338921},
  doi          = {10.1145/3338906.3338921},
  bibsource    = {dblp computer science bibliography, https://dblp.org}
}

@ARTICLE{ieee754,
  author={},
  journal={IEEE Std 754-2019 (Revision of IEEE 754-2008)}, 
  title={IEEE Standard for Floating-Point Arithmetic}, 
  year={2019},
  volume={},
  number={},
  pages={1-84},
  doi={10.1109/IEEESTD.2019.8766229}}

@article{Penrose_1955, title={A generalized inverse for matrices}, volume={51}, DOI={10.1017/S0305004100030401}, number={3}, journal={Mathematical Proceedings of the Cambridge Philosophical Society}, author={Penrose, R.}, year={1955}, pages={406–413}}

@book{DBLP:books/daglib/0086372,
  author       = {Gene H. Golub and
                  Charles F. Van Loan},
  title        = {Matrix Computations, Third Edition},
  publisher    = {Johns Hopkins University Press},
  year         = {1996},
  isbn         = {978-0-8018-5414-9},
  bibsource    = {dblp computer science bibliography, https://dblp.org}
}

@misc{SMTCOMP2025,
  author       = {{Martin Jon\'a\v{s} and Fran\c{c}ois Bobot and David D\'eharbe and Dominik Winterer}},
  title        = {{SMT-COMP 2025}: The 20th International Satisfiability Modulo Theories Competition},
  howpublished = {\url{https://smt-comp.github.io/2025/}},
  note         = {Accessed: 2025-11-12},
  year         = {2025}
}

@dataset{SMTLIB2025,
  author    = {Mathias Preiner and Hans-J{\"o}rg Schurr and Clark Barrett
               and Pascal Fontaine and Aina Niemetz and Cesare Tinelli},
  title     = {{SMT-LIB release 2025 (non-incremental benchmarks)}},
  year      = {2025},
  publisher = {Zenodo},
  note      = {SMT-LIB non-incremental benchmark library},
  url       = {https://zenodo.org/records/16740866}
}

@inproceedings{DBLP:conf/tacas/MouraB08,
author = {De Moura, Leonardo and Bj\o{}rner, Nikolaj},
title = {Z3: an efficient SMT solver},
year = {2008},
isbn = {3540787992},
publisher = {Springer-Verlag},
address = {Berlin, Heidelberg},
booktitle = {Proceedings of the Theory and Practice of Software, 14th International Conference on Tools and Algorithms for the Construction and Analysis of Systems},
pages = {337–340},
numpages = {4},
location = {Budapest, Hungary},
series = {TACAS'08/ETAPS'08}
}

@inproceedings{mathsat5,
  author = {Alessandro Cimatti and Alberto Griggio and Bastiaan Schaafsma and Roberto Sebastiani},
  title = {{The MathSAT5 SMT Solver}},
  editor = {Nir Piterman and Scott Smolka},
  booktitle = {Proceedings of TACAS},
  year = {2013},
  volume = {7795},
  series = {LNCS},
  publisher = {Springer},
}

@inproceedings{DBLP:conf/cav/NiemetzP23,
  author       = {Aina Niemetz and
                  Mathias Preiner},
  editor       = {Constantin Enea and
                  Akash Lal},
  title        = {Bitwuzla},
  booktitle    = {Computer Aided Verification - 35th International Conference, {CAV}
                  2023, Paris, France, July 17-22, 2023, Proceedings, Part {II}},
  series       = {Lecture Notes in Computer Science},
  volume       = {13965},
  pages        = {3--17},
  publisher    = {Springer},
  year         = {2023},
  url          = {https://doi.org/10.1007/978-3-031-37703-7\_1},
  doi          = {10.1007/978-3-031-37703-7\_1},
  bibsource    = {dblp computer science bibliography, https://dblp.org}
}

@inproceedings{DBLP:conf/tacas/BarbosaBBKLMMMN22,
author = {Barbosa, Haniel and Barrett, Clark and Brain, Martin and Kremer, Gereon and Lachnitt, Hanna and Mann, Makai and Mohamed, Abdalrhman and Mohamed, Mudathir and Niemetz, Aina and N\"{o}tzli, Andres and Ozdemir, Alex and Preiner, Mathias and Reynolds, Andrew and Sheng, Ying and Tinelli, Cesare and Zohar, Yoni},
title = {cvc5: A Versatile and Industrial-Strength SMT Solver},
year = {2022},
isbn = {978-3-030-99523-2},
publisher = {Springer-Verlag},
address = {Berlin, Heidelberg},
url = {https://doi.org/10.1007/978-3-030-99524-9_24},
doi = {10.1007/978-3-030-99524-9_24},
booktitle = {Tools and Algorithms for the Construction and Analysis of Systems: 28th International Conference, TACAS 2022, Held as Part of the European Joint Conferences on Theory and Practice of Software, ETAPS 2022, Munich, Germany, April 2–7, 2022, Proceedings, Part I},
pages = {415–442},
numpages = {28},
location = {Munich, Germany}
}

@inproceedings{weak-distance,
author = {Fu, Zhoulai and Su, Zhendong},
title = {Effective floating-point analysis via weak-distance minimization},
year = {2019},
isbn = {9781450367127},
publisher = {Association for Computing Machinery},
address = {New York, NY, USA},
url = {https://doi.org/10.1145/3314221.3314632},
doi = {10.1145/3314221.3314632},
booktitle = {Proceedings of the 40th ACM SIGPLAN Conference on Programming Language Design and Implementation},
pages = {439–452},
numpages = {14},
location = {Phoenix, AZ, USA},
series = {PLDI 2019}
}

@misc{grater-experiment,
  author       = {Qian Chen and Chenqi Cui and Fengjuan Gao and Yu Wang
                  and Ke Wang and Linzhang Wang},
  title        = {Grater-experiment: artifact and code for
                  ``Solving Floating-Point Constraints with
                  Continuous Optimization''},
  howpublished = {\url{https://github.com/grater-exp/grater-experiment}},
  year         = {2025},
  note         = {Commit \texttt{e405648}, accessed 13 Nov 2025}
}

@inproceedings{DBLP:conf/cav/dpll1,
  author       = {Harald Ganzinger and
                  George Hagen and
                  Robert Nieuwenhuis and
                  Albert Oliveras and
                  Cesare Tinelli},
  editor       = {Rajeev Alur and
                  Doron A. Peled},
  title        = {{DPLL(} {T):} Fast Decision Procedures},
  booktitle    = {Computer Aided Verification, 16th International Conference, {CAV}
                  2004, Boston, MA, USA, July 13-17, 2004, Proceedings},
  series       = {Lecture Notes in Computer Science},
  volume       = {3114},
  pages        = {175--188},
  publisher    = {Springer},
  year         = {2004},
  url          = {https://doi.org/10.1007/978-3-540-27813-9\_14},
  doi          = {10.1007/978-3-540-27813-9\_14},
  bibsource    = {dblp computer science bibliography, https://dblp.org}
}

@article{DBLP:journals/jacm/dpll2,
  author       = {Robert Nieuwenhuis and
                  Albert Oliveras and
                  Cesare Tinelli},
  title        = {Solving {SAT} and {SAT} Modulo Theories: From an abstract Davis--Putnam--Logemann--Loveland
                  procedure to DPLL(\emph{T})},
  journal      = {J. {ACM}},
  volume       = {53},
  number       = {6},
  pages        = {937--977},
  year         = {2006},
  url          = {https://doi.org/10.1145/1217856.1217859},
  doi          = {10.1145/1217856.1217859},
  bibsource    = {dblp computer science bibliography, https://dblp.org}
}

@article{dpll3,
author = {Davis, Martin and Putnam, Hilary},
title = {A Computing Procedure for Quantification Theory},
year = {1960},
issue_date = {July 1960},
publisher = {Association for Computing Machinery},
address = {New York, NY, USA},
volume = {7},
number = {3},
issn = {0004-5411},
url = {https://doi.org/10.1145/321033.321034},
doi = {10.1145/321033.321034},
journal = {J. ACM},
month = jul,
pages = {201–215},
numpages = {15}
}

@ARTICLE{cdcl1,
  author={Marques-Silva, J.P. and Sakallah, K.A.},
  journal={IEEE Transactions on Computers}, 
  title={GRASP: a search algorithm for propositional satisfiability}, 
  year={1999},
  volume={48},
  number={5},
  pages={506-521},
  doi={10.1109/12.769433}}

@ARTICLE{aflfast,
  author={Böhme, Marcel and Pham, Van-Thuan and Roychoudhury, Abhik},
  journal={IEEE Transactions on Software Engineering}, 
  title={Coverage-Based Greybox Fuzzing as Markov Chain}, 
  year={2019},
  volume={45},
  number={5},
  pages={489-506},
  doi={10.1109/TSE.2017.2785841}}

@misc{afl,
  author       = {Micha{\l} Zalewski},
  title        = {Technical whitepaper for afl-fuzz},
  howpublished = {\url{http://lcamtuf.coredump.cx/afl/technical_details.txt}},
  note         = {Accessed 13 Nov 2025}
}

@inproceedings{mathsat-fmcad,
  author       = {Leopold Haller and
                  Alberto Griggio and
                  Martin Brain and
                  Daniel Kroening},
  editor       = {Gianpiero Cabodi and
                  Satnam Singh},
  title        = {Deciding floating-point logic with systematic abstraction},
  booktitle    = {Formal Methods in Computer-Aided Design, {FMCAD} 2012, Cambridge,
                  UK, October 22-25, 2012},
  pages        = {131--140},
  publisher    = {{IEEE}},
  year         = {2012},
  url          = {https://ieeexplore.ieee.org/document/6462565/},
  bibsource    = {dblp computer science bibliography, https://dblp.org}
}

@inproceedings{DBLP:conf/aaai/FrohlichBWH15,
  author       = {Andreas Fr{\"{o}}hlich and
                  Armin Biere and
                  Christoph M. Wintersteiger and
                  Youssef Hamadi},
  editor       = {Blai Bonet and
                  Sven Koenig},
  title        = {Stochastic Local Search for Satisfiability Modulo Theories},
  booktitle    = {Proceedings of the Twenty-Ninth {AAAI} Conference on Artificial Intelligence,
                  January 25-30, 2015, Austin, Texas, {USA}},
  pages        = {1136--1143},
  publisher    = {{AAAI} Press},
  year         = {2015},
  url          = {https://doi.org/10.1609/aaai.v29i1.9372},
  doi          = {10.1609/AAAI.V29I1.9372},
  bibsource    = {dblp computer science bibliography, https://dblp.org}
}

@ARTICLE{numerical,
  author={Miller, W. and Spooner, D.L.},
  journal={IEEE Transactions on Software Engineering}, 
  title={Automatic Generation of Floating-Point Test Data}, 
  year={1976},
  volume={SE-2},
  number={3},
  pages={223-226},
  doi={10.1109/TSE.1976.233818}}

@book{numerical2,
  author       = {Jorge Nocedal and
                  Stephen J. Wright},
  title        = {Numerical Optimization},
  publisher    = {Springer},
  year         = {1999},
  url          = {https://doi.org/10.1007/b98874},
  doi          = {10.1007/B98874},
  isbn         = {978-0-387-98793-4},
  bibsource    = {dblp computer science bibliography, https://dblp.org}
}

@misc{z3py-tutorial,
  author       = {Leonardo de Moura and Nikolaj Bj{\o}rner and Eric Pony},
  title        = {Z3Py Tutorial},
  year         = {2012},
  howpublished = {\url{https://ericpony.github.io/z3py-tutorial/guide-examples.htm}},
  note         = {Accessed 13 November 2025}
}

@article{ulp1991,
author = {Goldberg, David},
title = {What every computer scientist should know about floating-point arithmetic},
year = {1991},
issue_date = {March 1991},
publisher = {Association for Computing Machinery},
address = {New York, NY, USA},
volume = {23},
number = {1},
issn = {0360-0300},
url = {https://doi.org/10.1145/103162.103163},
doi = {10.1145/103162.103163},
journal = {ACM Comput. Surv.},
month = mar,
pages = {5–48},
numpages = {44}
}

@article{parsat,
  author       = {Markus Krahl and
                  Matthias G{\"{u}}demann and
                  Stefan Wallentowitz},
  title        = {parSAT: Parallel Solving of Floating-Point Satisfiability},
  journal      = {CoRR},
  volume       = {abs/2509.16237},
  year         = {2025},
  url          = {https://doi.org/10.48550/arXiv.2509.16237},
  doi          = {10.48550/ARXIV.2509.16237},
  eprinttype    = {arXiv},
  eprint       = {2509.16237},
  bibsource    = {dblp computer science bibliography, https://dblp.org}
}

@article{basinhopping,
   title={Global Optimization by Basin-Hopping and the Lowest Energy Structures of Lennard-Jones Clusters Containing up to 110 Atoms},
   volume={101},
   ISSN={1520-5215},
   url={http://dx.doi.org/10.1021/jp970984n},
   DOI={10.1021/jp970984n},
   number={28},
   journal={The Journal of Physical Chemistry A},
   publisher={American Chemical Society (ACS)},
   author={Wales, David J. and Doye, Jonathan P. K.},
   year={1997},
   month=jul, pages={5111–5116} }

@ARTICLE{isres,
  author={Runarsson, T.P. and Xin Yao},
  journal={IEEE Transactions on Systems, Man, and Cybernetics, Part C (Applications and Reviews)}, 
  title={Search biases in constrained evolutionary optimization}, 
  year={2005},
  volume={35},
  number={2},
  pages={233-243},
  doi={10.1109/TSMCC.2004.841906}}

@article{crs2,
  author  = {P. Kaelo and M. M. Ali},
  title   = {Some Variants of the Controlled Random Search Algorithm for Global Optimization},
  journal = {Journal of Optimization Theory and Applications},
  volume  = {130},
  number  = {2},
  pages   = {253--264},
  year    = {2006},
  doi     = {10.1007/s10957-006-9101-0}
}

\clearpage
\appendix
\section{StageSat Implementation Pseudocode}
\label{sec:pseudocode}

\begin{algorithm}
\caption{Solve Objective Function with Global Minimization}
\begin{algorithmic}[1]
\State \textbf{Input:}
\State \quad $C$: Quantifier-free floating-point constraints of FP
\State \quad $R_{sq}$: Fast projection-aided objective function
\State \quad $R_{ulp}$: ULP distance objective function
\State \textbf{Output:}
\State \quad \textsc{sat} with a model of $C$, if found, or \textsc{unsat-guess} otherwise
\State
\Procedure{Global\_Minimizer}{$C, R_{sq}, R_{ulp}$}
    \For{$i = 1$ \textbf{to} $nStartOver$}
        \State $\triangleright$ Round 1: Fast convergence using $R_{sq}$
        \State $x_1^{(0)} \gets$ \Call{Random\_Starting\_Point()}{}
        \State $x_1^{(L)} \gets$ \Call{Local\_Minimize}{$R_{sq}, x_1^{(0)}$}
        \State $\triangleright$ Round 2: Refine using precise $R_{ulp}$
        \State $x_2^{(L)} \gets$ \Call{Local\_Minimize}{$R_{ulp}, x_1^{(L)}$}
        \If{\Call{IsModel}{$C, x_2^{(L)}$}}
            \Comment{Use Z3 library function for model validation}
            \State \Return \textsc{sat}, $x_2^{(L)}$
        \EndIf
        \State $\triangleright$ Round 3: Fine-grained search in FP neighborhood
        \State $N^{(0)} \gets$ \Call{zeros}{dim} \Comment{Start from zero offsets}
        \State Define $f(N) = R_{ulp}([\Call{Nth\_Floating\_Point}{N[j], x_2^{(L)}[j]} \text{ for } j = 1 \text{ to } dim])$
        \State $N^{(L)} \gets$ \Call{Local\_Minimize}{$f, N^{(0)}$}
        \State $x_3^{(L)} \gets [\Call{Nth\_Floating\_Point}{N^{(L)}[j], x_2^{(L)}[j]} \text{ for } j = 1 \text{ to } dim]$
        \If{\Call{IsModel}{$C, x_3^{(L)}$}}
            \State \Return \textsc{sat}, $x_3^{(L)}$
        \EndIf
    \EndFor
    \State \Return ``\textsc{unsat-guess}''
\EndProcedure
\State
\Function{Local\_Minimize}{$R, x^{(0)}$}
    \State $x^{(L)} \gets$ scipy.optimize.basinhopping($R, x^{(0)}$)
    \State \Return $x^{(L)}$
\EndFunction
\State
\Function{Nth\_Floating\_Point}{$n, x$}
    \State $\triangleright$ Time complexity: $O(1)$
    \State \Return the $n$-th representable floating-point number from $x$
\EndFunction
\end{algorithmic}
\end{algorithm}

\begin{algorithm}
\caption{Generate Objective Function}
\begin{algorithmic}[1]
\State \textbf{Input:} $C$ (quantifier-free FP constraints)
\State \textbf{Output:} $R_{sq}$, $R_{ulp}$
\State
\Procedure{Generate\_R\_Square}{$C$}
    \State $C_{leq}, C_{other} \gets \Call{Parse\_C}{C}$
    \If{$C_{leq} = \emptyset$}
        \State $R_{sq}(x) \gets \sum C_{other}$
    \Else
        \State $A,b \gets \Call{Build\_Matrix}{C_{leq}}$
        \State $G^{+} \gets \Call{MoorePenrose}{A A^{\top}}$
        \State $P \gets A^{\top} G^{+}$;\quad $M \gets I - P A$;\quad $c \gets P b$
        \State $R_{sq}(x) \gets \sum_{i=1}^{n} (x_i - (M x + c)_i)^2 + \sum C_{other}$
    \EndIf
    \State \Return $R_{sq}$
\EndProcedure
\State
\Procedure{Generate\_R\_ULP}{$C$}
    \State $C = \{c_i : \text{lhs}_i = \text{rhs}_i \mid i = 1,\ldots,m\}$
    \State $R_{ulp}(x) \gets \sum_{i=1}^{m} \text{ulp}(\text{lhs}_i, \text{rhs}_i)$
    \State \Return $R_{ulp}$
\EndProcedure
\end{algorithmic}
\end{algorithm}

\end{document}